\newif\ifconf
\conftrue

\newif\ifcomm
\commfalse

\newif\ifblind
\blindtrue

\newif\ifs
\sfalse

\documentclass[sigconf]{acmart}

\usepackage[english]{babel}
\usepackage{blindtext}
\usepackage{subcaption}
\renewcommand\footnotetextcopyrightpermission[1]{} %
\setcopyright{none}

\acmDOI{}

\acmISBN{}

\acmConference[]%
\acmYear{}%

\acmPrice{}

\usepackage{graphicx}
\usepackage{xurl}
\usepackage[normalem]{ulem} %
\usepackage{xcolor} %

\usepackage{xspace}
\usepackage{balance}
\usepackage{booktabs}  %
	
\usepackage{array}
	\newcolumntype{C}[1]{>{\centering\let\newline\\\arraybackslash\hspace{0pt}}m{#1}}

\usepackage[inline,shortlabels]{enumitem} %
\newcommand{\bl}{\begin{enumerate*}[(1)]}
\newcommand{\el}{\end{enumerate*}\xspace}

\newtheoremstyle{boldthm}{}{}{\itshape}{}{\bfseries}{.}{ }{\thmname{#1}\thmnumber{ #2}\thmnote{ (#3)}} %
\theoremstyle{boldthm}

\newtheorem{theorem}{Theorem}%

\newcommand{\bp}{\begin{proof}}
\newcommand{\bpo}{ \begin{proof}[Proof Outline] }
\newcommand{\ep}{\end{proof}}       %
\newcommand{\proofof}[1]{\begin{proof}[Proof of #1]} %

\usepackage[capitalize,noabbrev]{cleveref}
\crefname{equation}{Eq.}{Eqs.}
\Crefname{equation}{Eq.}{Eqs.}
\crefname{figure}{Fig.}{Figs.}
\Crefname{figure}{Fig.}{Figs.}
\crefname{table}{Table}{Tables}
\Crefname{table}{Table}{Tables}
\crefname{property}{Property}{Properties}
\Crefname{property}{Property}{Properties}
\crefname{section}{\S}{Secs.}
\Crefname{section}{\S}{Secs.}
\crefname{subsection}{\S}{Secs.}
\Crefname{subsection}{\S}{Secs.}
\crefname{appendix}{Appendix}{Appendices}
\Crefname{appendix}{Appendix}{Appendices}

\ifcomm
  \newcommand{\mycomm}[3]{{\footnotesize{{\color{#2} \textbf{[#1: #3]}}}}}
  \newcommand{\Fmycomm}[3]{{\color{red} \footnote{{{\color{#2} \textbf{[#1: #3]}}} }}}
\else
    \newcommand{\mycomm}[3]{}
    \newcommand{\Fmycomm}[3]{}
\fi

\ifs
	\newcommand{\SmallReduceVSpace}{\vspace{-0pt}}	
	\newcommand{\ReduceVSpace}{\vspace{-0pt}}

    \setlist{leftmargin=*} %
    \newcommand{\T}[1]{\par\smallskip\noindent\textbf{#1}} %
    \newcommand{\Ts}[1]{\par\vspace{1pt plus 1pt minus 1pt}\noindent\textit{#1}} %
\else
	
	\newcommand{\SmallReduceVSpace}{}	
	\newcommand{\ReduceVSpace}{}

    \newcommand{\T}[1]{\par\smallskip\noindent\textbf{#1}} %
    \newcommand{\Ts}[1]{\par\smallskip\noindent\textit{#1}} %
\fi

\newcommand{\be}{\begin{equation}}
\newcommand{\ee}{\end{equation}}

\newcommand{\para}[1]{\left( #1 \right)}        %
\newcommand{\brac}[1]{\left\{ #1 \right\}}
\newcommand{\sbrac}[1]{\left[ #1 \right]}

\newcommand{\N}{\mathbb{N}}

\newcommand{\unit}[1]{\;\mathrm{#1}} %

\newcommand{\vx}{\checkmark\kern-1.1ex\raisebox{.7ex}{\rotatebox[origin=c]{125}{--}}} %

\usepackage{pifont}
\usepackage{arydshln} %
 
\providecommand{\vs}{{vs.}\xspace}
\providecommand{\ie}{{i.e.,}\xspace}
\providecommand{\eg}{{e.g.,}\xspace}

\providecommand{\etc}{{etc.}\xspace}

\newcommand{\newVar}[2]{\newcommand{#1}{\ensuremath{#2}\xspace}}
\newcommand{\renewVar}[2]{\renewcommand{#1}{\ensuremath{#2}\xspace}}
  \newVar{\server}{S}
  \newVar{\client}{C}
  \newVar{\rclient}{R_c}
  \newVar{\rserver}{R_s}

\newcommand{\DSM}{D^\text{\textsc{SM}}(t)}
\newcommand{\DHBM}{D^\text{\textsc{HBM}}(t)}
\newcommand{\QSM}{Q^\text{\textsc{SM}}(t)}
\newcommand{\QHBM}{Q^\text{\textsc{HBM}}(t)}

\newcommand{\name}{PFI\xspace}
\newcommand{\arch}{SPS\xspace}
\newcommand{\HBM}{HBM switch\xspace}
\newcommand{\HBMs}{HBM switches\xspace}

\newcommand{\batch}{batch\xspace}
\newcommand{\batches}{batches\xspace}

\renewVar{\k}{k}

\newVar{\K}{K}
\renewVar{\N}{N}
\renewVar{\j}{j}
\renewVar{\th}{^\text{th}}

\begin{document}
\title{Scaling Routers with In-Package Optics and High-Bandwidth Memories}

\newcommand{\aut}[2]{#1\texorpdfstring{$^{#2}$}{(#2)}}  %
\author{
  \aut{Isaac Keslassy}{1,2}, 
  \aut{Ilay Yavlovich}{1}, 
  \aut{Jose Yallouz}{1}, \\
  \aut{Tzu-Chien Hsueh}{3}, 
  \aut{Yeshaiahu Fainman}{3}, 
  \aut{Bill Lin}{3}
}%
\affiliation{
$^1$ \textit{Technion} \quad 
$^2$ \textit{UC Berkeley}\quad 
$^3$ \textit{UC San Diego} }
\renewcommand{\shortauthors}{Keslassy et al.}

\begin{abstract}
This paper aims to apply two major scaling transformations from the computing packaging industry to internet routers: the heterogeneous integration of high-bandwidth memories (HBMs) and chiplets, as well as in-package optics. 
We propose a novel internet router architecture that employs these technologies to achieve a petabit/sec router within a single integrated package. At the top-level, we introduce a novel split-parallel switch architecture that spatially divides (without processing) the incoming fibers and distributes them across smaller independent switches without intermediate OEO conversions or fine-tuned per-packet load-balancing. This passive spatial division enables scaling at the cost of a coarser traffic load balancing. Yet, through extensive evaluations of backbone network traffic, we demonstrate that differences with fine-tuned approaches are small. 
In addition, we propose a novel HBM-based shared-memory architecture for the implementation of the smaller independent switches, and we introduce a novel parallel frame interleaving algorithm that packs traffic into frames so that HBM banks are accessed at peak HBM data rates in a cyclical interleaving manner. We further discuss why these new technologies represent a paradigm shift in the design of future internet routers. Finally, we emphasize that power consumption may constitute the primary bottleneck to scaling.
\end{abstract}

\maketitle
\sloppypar %
\section{Introduction}

\textit{The goal of this paper is to apply two major scaling  transformations from the computing packaging industry to internet routers, \ie heterogeneous integration of HBMs and chiplets, and in-package optics,
by introducing a novel router architecture that can harness these emerging technologies.} 

The computing industry is undergoing a profound shift driven by the slowing of traditional scaling laws and the rapid growth of data-intensive applications. The long-standing trends of Moore’s law and Dennard scaling no longer sustain exponential gains in transistor density and power efficiency, creating pressing challenges for future system performance. Simultaneously, the rise of generative AI and other data-centric workloads has heightened the demand for massive compute and memory. To address these pressures, the semiconductor industry has embarked on two transformative directions: \textit{heterogeneous integration within packages} and \textit{in-package optical interconnects}.

Heterogeneous integration enables the co-packaging of multiple \textit{processing chiplets} and three-dimensionally stacked \textit{High Bandwidth Memory (HBM)} modules~\cite{HBM, HBM4} within a single package. Recent HBM4 stacks~\cite{jedec,jedec-tom,skhynix,skhynix-ee}  provide up to $64\unit{ GB}$ of memory capacity and $20.48\unit{ Tb/s}$ of bandwidth, allowing multi-stack assemblies to reach multi-terabyte memory capacities and petabit-scale memory bandwidth within a compact footprint~\cite{HBM-micron}. Panel-scale substrates can further accommodate many such chiplets and HBM modules, significantly increasing the compute and memory density per package~\cite{panel500}.

In parallel, \textit{optical package I/O} and \textit{optical interconnects within the package}, referred together as \textit{in-package optics},
have emerged as key enablers
for scaling communication bandwidth~\cite{LM1, LM2, LM3}. Current photonics packaging supports over $100\unit{Tb/s}$ of optical I/O, with technology roadmaps targeting an order-of-magnitude increase through wider fiber ribbons~\cite{LM1} and higher wavelength multiplexing density~\cite{Rizzo22,Deng24}. 

Together, these emerging technologies enable package-level systems with unprecedented compute, memory, and communication bandwidth densities. They motivate a fundamental rethinking of internet router architectures. This paper proposes a novel router-in-a-package design that exploits HBM memory, processing chiplets and in-package optics to implement a petabit-per-second router within a single integrated package. 

\T{Our architecture introduces two key contributions:}

\T{(1)~Split-Parallel Switch  (\arch)}, a router architecture that spatially splits (without processing) incoming optical fibers across multiple smaller, independent switches. This passive but coarse load-balancing approach exploits the fact that data arrives as optical signals to enable efficient high-throughput scaling, without the many Optical-to-Electrical (OE) and Electrical-to-Optical (EO) conversions that limit fine-tuned electronic load balancing~\cite{LBR,rotornet,opera,PPS}, and without the I/O capacity and power overhead that limits mesh-based architectures~\cite{U2TURN, isca24}.

\T{(2)~HBM Switch}, a novel shared-memory design for the smaller switches that leverages HBM stacks to mimic an ideal output-queued shared-memory switch with a small speedup. There is no known design that provides such a performance guarantee at these speeds. These \HBMs differ from prior shared-memory switch architectures~\cite{fast-shared-memory-switches,sms1,sms2,sms3} in that they take advantage of the capabilities of modern HBMs to provide significant switching and buffering capabilities. We propose a \textit{Parallel Frame Interleaving (PFI)} mechanism that (i)~aggregates packets into wide frames to exploit HBM bandwidth, (ii)~cyclically interleaves memory bank accesses for peak utilization, and (iii)~employs cyclical crossbars to avoid any scheduling.

Combining these contributions leads to a potential increase in  router capacity per area by 1-2 orders of magnitude \vs current routers. 

\T{Paper outline.}
\cref{sec:top-parallel} and \cref{sec:top-HBM} present our thought process and design choices behind the SPS and HBM switch architectures, respectively.
\cref{sec:analysis} provides detailed analysis of the proposed designs, and
\cref{sec:guarantees} provides their performance guarantees.
Given our proposed approach is based on passive spatial splitting of fibers into independent smaller and slower switches, which 
constitutes a coarse approach to traffic load balancing, we want to make sure that the approach does not lead to significant congestion or losses. 
\cref{sec:eval} provides
extensive evaluations of backbone network traffic and synthetic cross-datacenter AI workload traffic to demonstrate that differences with fine-tuned approaches are indeed negligible. 
\cref{sec:discussion} discusses the impact of this work on networking future.
We explain why these new technologies and our proposed architecture form a paradigm shift in how we should conceive and design internet routers, including how they impact router buffer sizing and management. 
Finally, \cref{sec:conclusion} concludes and looks to the road ahead. We argue that power consumption
will become the primary scaling limiter for future router implementations.

\section{Top-Level Split-Parallel Switch}
\label{sec:top-parallel}
\subsection{Design process}

We start by describing the thought process that led to our {\em Split-Parallel Switch (\arch)} design. We mention the leading alternatives in the literature and explain why they do not fit our goals of providing a high throughput with a reasonable power consumption.

\T{Design 1.} The simplest design for an $N \times N$ router consists of \textit{a single centralized switch fabric}. 

\T{Challenge 1.} 
A single centralized switch cannot keep up with our needed high rates, as it would need prohibitive switching rates as well as memory access rates.

\T{Idea 1.} In an $N \times N$ router, we want to introduce $H$ smaller switches to decompose the switching work of our large router into $H$ smaller conceptual routers.

\T{Design 2.} A standard approach with a large number $H$ of smaller switches is to organize them in a $\sqrt{H} \times \sqrt{H}$ two-dimensional \textit{mesh} 
configuration, \ie each smaller switch is connected to its adjacent neighbors in the same row or column.
Incoming packets enter some smaller switch, and pass through a number of intermediate smaller switches %
before exiting.

\T{Challenge 2.} Passing through the many intermediate hops wastes a significant portion of the link and switch capacity, and power. For instance, in a $10\times 10$ mesh, the guaranteed capacity is at most 20\% of the total capacity for an arbitrary admissible traffic pattern, wasting 80\% of the capacity and power~\cite{U2TURN}. A recent work~\cite{isca24} proposes to add instead extra I/Os to the switch chiplets %
to implement the pass-throughs. Still, the number of required extra I/Os would be substantial.

\T{Idea 2.} We want to make sure that each packet goes through a number of hops that does not grow with $H$.

\T{Design 3.} We can adopt a \textit{three-stage} network configuration, such as a \textit{Clos} network with three stages of $H/3$.
Using this design, we could schedule packets using a {load-balanced router} approach or another demand-oblivious schedule~\cite{LBR,rotornet,opera}, or a parallel packet switch architecture~\cite{PPS}.

\T{Challenge 3.} These approaches assume that each input can adopt a packet-by-packet load balancing, and each output can implement a reordering buffer. 
These are near-impossible to implement in optics, so in practice, all three stages will need electronics. 
In addition, within the package, we cannot organize all the smaller switches next to each other, so we will also need {additional} optical waveguide-based links between consecutive logical switches.
Thus, a three-stage architecture would need three electronic stages separated by optics, \ie three OEO conversion stages that would increase the power consumption,
{both because of the conversions themselves and because the three stages split the processing and memory functions into more chiplets.}

\T{Idea 3.} Organize the $H$ switches such that a single OEO conversion is needed, \ie make sure that all the electronic processing for a packet is confined within a single switch. %

\T{Design 4.} We propose to organize our $H$ smaller switches as \textit{$H$ parallel switches}, such that any incoming packet only accesses \textit{a single} smaller switch and goes through a single OEO conversion. %
We cannot use electronics to load-balance packets among the $H$ stages, because it would add an OEO conversion. Thus, we rely on \textit{a poor man's solution} of evenly splitting the incoming fibers at each input. We assign $1/H$-th of the fibers to each switch. That is, given $F$ fibers per input, we split the fibers in a straightforward manner by connecting the first $F/H$ fibers that enter each input to the first switch, \etc %
This solution is non-optimal and does {not} guarantee a perfect load-balancing as in electronic per-packet load-balancing. 
{In fact, the uneven distribution across smaller switches operating at a reduced capacity may potentially lead to packet losses.}
This is the cost of having a single OEO stage.

\T{Challenge 4.} Our architecture satisfies our power goals, but the straightforward fiber splitting has a couple of 
issues:
(1)~Practically, the first fiber of each input is typically connected first, and therefore has a higher load. Thus, the above splitting pattern typically causes a higher load in the first switch, yielding a poorer load-balancing. 
(2)~In addition, an adversarial attacker could exploit the known internal splitting pattern of the fibers even without knowing the specific internal aspects of the router.

\T{Idea 4.} We use a pseudo-random pattern to pick the $F/H$ fibers that connect each input to each middle switch.

\subsection{Architecture}
\label{sec:parallel}

In this section, we describe the \arch architecture. We parameterize the dimensions so it can be scaled to different capacities, and apply a set of implementation parameters that achieve a petabit total I/O in a single package.

\begin{figure}
    \centering
    \includegraphics[width=1.0\linewidth]{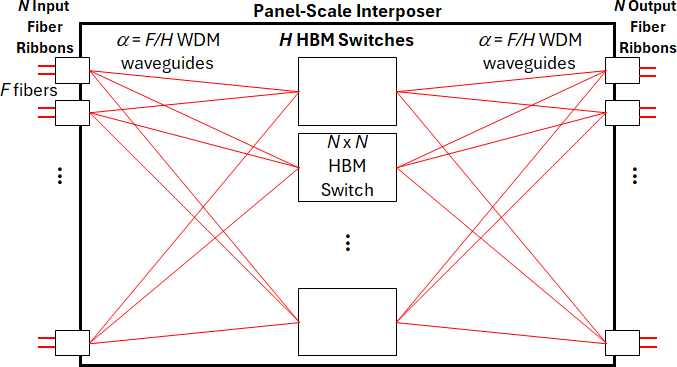}
    \caption{The
    \arch architecture.
    }
    \label{fig:parallel}
    \ReduceVSpace %
\end{figure}

\T{Modules.} \Cref{fig:parallel} illustrates the \arch architecture that fits on a single package. The package includes $N = 16$ 
fiber-ribbon arrays. Each such array comprises $F = 64$ fibers. 
Thus, there are $N \cdot F = 1,024$ 
fibers per package. 

Each fiber carries $W = 16$ {\em input} wavelength-division-multiplexing (WDM) channels, of bandwidth 
{$R = 40$\,Gb/s} each. 
Likewise, for better packaging, each of the {\em same} $N = 16$ fiber ribbons also serves as the {\em egress} of the package, with each fiber also carrying a separate set of $W = 16$ wavelengths as {\em output} WDM channels. 
Overall, a package offers {$N \cdot F \cdot W \cdot R = 655$\,Tb/s} of I/O in each direction, \ie \textit{a total I/O of} $1.31$\,Pb/s.

Within the package, there are $H = 16$ parallel $N \times N$ \HBMs. Each such switch must support $1/H=1/16\th$ of the total I/O, \ie $2 (N \cdot F \cdot W \cdot R)/H = 81.92\unit{Tb/s}$ of memory I/O. The \HBMs include several HBM stacks (\cref{sec:HBM}).

At each fiber ribbon, 
wavelengths coming through the $F$ optical fibers are \textit{passively} coupled to the corresponding {wavelengths in the $F$} internal WDM waveguides 
via optical couplers for continuing the optical signal propagation within the package.
From each 
ribbon, a pseudo-random set of $\alpha = F / H = 4$ 
WDM waveguides is connected to each of the $H = 16$ smaller \HBMs. 
Likewise, at the egress, $\alpha = F / H = 4$ WDM waveguides are received from each \HBM at each output fiber ribbon.

\T{Operation.} All traffic incoming at some fiber of some input fiber ribbon is propagated through
its corresponding waveguide, and arrives at one of the \HBMs. There, as detailed in \cref{sec:HBM}, the \HBMs convert the optical signal into  an electronic signal,
extract the packets, and switch them to the \HBM egress corresponding to their output fiber ribbon. They are then reconverted into an optical signal, sent in one of the waveguide channels for this output ribbon, and 
forwarded through the corresponding egress fiber. 

\T{Packaging.} \Cref{fig:package} shows a possible layout within a
2.5D photonics interposer. The $N = 16$ fiber ribbons could be organized as 4 arrays on each side, 
while the $H = 16$ \HBMs could be arranged as a $4 \times 4$ matrix in the middle. For clarity, only the WDM waveguides connected to the first fiber ribbon are outlined. %

\begin{figure}
    \centering
    \includegraphics[width=0.9\linewidth]{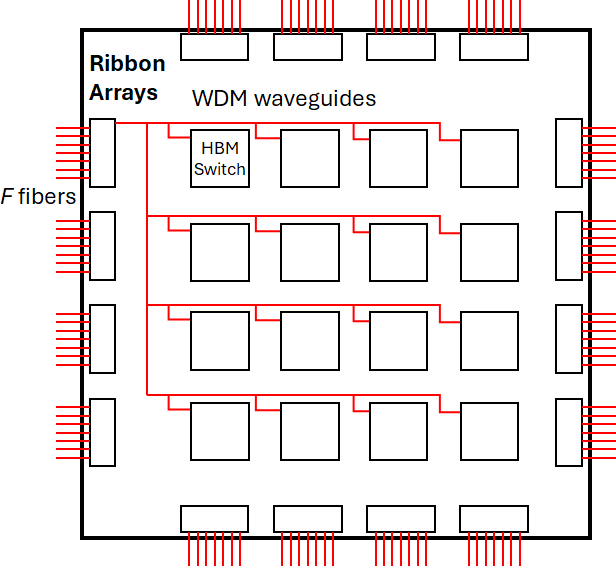}
    \caption{Packaging view of \arch on a 2.5D photonics interposer.}
    \label{fig:package}
    \ReduceVSpace %
\end{figure}

\T{Modularity.} 
{The \arch architecture enables a modular approach, from a single dense $1.31\unit{Pb/s}$ I/O package with 16 HBM switches, to 16 parallel packages of 
$1/16\th$ the capacity.}

\section{HBM Switch}
\label{sec:top-HBM}

\subsection{Design process}
\label{sec:HBM-process}
We now describe the design of the smaller \HBMs.
\T{Challenge 5.} 
Although the SPS architecture introduced in \cref{sec:top-parallel} splits incoming traffic across $H$ smaller switches operating in parallel, each $N \times N$ smaller switch must still handle a significant load of $1/H=1/16\th$ of the total traffic. 
This translates into a total incoming and outgoing traffic of  $2 (N \cdot F \cdot W \cdot R)/H = 81.92\unit{Tb/s}$.

\T{Idea 5.} To support this large combined data rate, we propose to use modern HBMs. %

\T{Design 5.} Our reference design, called the \textit{\HBM}, uses the HBM4 technology~\cite{jedec,jedec-tom} that provides a $2,048$-bit ultra-wide interface organized as $32$ channels of $64$ bits each. 
Recently announced HBM4 commercial offerings~\cite{micron,skhynix,skhynix-ee} have exceeded the HBM4 standard by achieving over $10\unit{Gb/s}$ per bit for a total bandwidth of $2,048 \cdot 10\unit{Gb/s} = 20.48\unit{Tb/s}$. 
By grouping $B=4$ HBM stacks, we obtain $4 \cdot 20.48\unit{Tb/s}=81.92\unit{Tb/s}$ of total bandwidth through $T=4 \cdot 32=128$ parallel channels.

\T{Challenge 6.} To use HBMs at peak data rate, memory accesses need to satisfy complex timing rules for accessing banks, rows, columns, \etc It is hard to reconcile (1)~the arbitrary and non-uniform traffic patterns that the \HBM needs to service with (2)~these complex HBM access rules, while (3)~providing strong and predictable performance guarantees. 
There are two main approaches in the literature. One is to provide a very strong theoretical guarantee and fully implement an ideal output-queued (OQ) shared-memory switch~\cite{fast-shared-memory-switches,sms1,sms2,sms3}. However, given current large HBMs, tracking packet locations (``bookkeeping'') would require prohibitive SRAM sizes of several GBs, and simulating a shadow ideal switch would incur prohibitive compute requirements. Another approach is to randomly spread packets across memory modules~\cite{Shrimali2005PacketBuffers}, then use a large reordering buffer at the output~\cite{yang2017simple,usubutun2023switches,rottenstreich2012switch}. This is akin to packet spraying in datacenters~\cite{RPS,LTCP,CAPS,Stellar}.

Unfortunately, both approaches are oblivious to the specific HBM memory rules and assume \textit{worst-case} random access times, leading to a large throughput loss. Specifically, this worst-case random access needs about $30\unit{ns}$~\cite{jedec} just to activate and close (precharge) banks. 
Giving these approaches the benefit of the doubt and assuming that they can leverage the HBM parallel channels, which did not exist 20 years ago, they would still suffer from throughput reduction factors ranging from $2.6\times$ for $1,500$-byte packets to $39\times$ for worst-case $64$-byte ones. If they don't leverage parallel channels, the reduction can reach $1,250\times$.

\T{Idea 6.} 
We aim to aggressively design our switch to operate at the \textit{best-case} memory access times rather than the worst-case.
To achieve peak data rates, we need to write at once a large quantity of data to harness the many parallel channels available via the HBM's ultra-wide interface. We also want this data to have the same output so it can be read and depart together. To do so, we propose to aggregate incoming packets with a shared output destination into \textit{frames} that can be efficiently striped across all HBM channels in parallel. These frames are similar to a shuttle that waits for passengers to a given destination and departs as soon as it is packed. Frames can contain packets from different inputs as long as they share the same output destination, offering flexible aggregation. These frames are inspired by past works on memory management in single linecards~\cite{iyer2008} and the uniform frame spreading (UFS) algorithm in load-balanced routers~\cite{thesis}.

\T{Design 6.} We devise a \textit{Parallel Frame Interleaving (\name)} algorithm. We outline it next %
before detailing it in~\cref{sec:HBM}.

\Ts{(1) Frame aggregation.}
\name aggregates traffic in two stages. First, at each input, variable-size packets arrive at per-output queues, where they are cut and assembled into fixed-size batches of $k=4\unit{KB}$. Then, at an intermediate buffering stage, batches enter per-output queues, where they are assembled into $K=512\unit{KB}$ 
frames of $K/k=128$ batches.

\Ts{(2) Slicing.}
To prepare for the highly parallel HBM memory accesses, \name implements the intermediate buffering stage by slicing each batch into $N$ equal chunks that are placed in $N$ parallel intermediate memory modules. 
Inspired by prior load-balanced switch architectures~\cite{thesis, LBR,yang2017safe,Lin2010Concurrent} that do not require switch-fabric scheduling, \name uses  
an $N \times N$ cyclical crossbar between the $N$ inputs and the $N$ memory modules. The crossbar staggers the forwarding of the $N$ batch slices from each input to the $N$ intermediate memory modules. Thus, each frame of 128 batches is also sliced into $N$ parts.

\Ts{(3) Bank interleaving.}
\name uses \textit{predictable} and \textit{staggered} bank accesses to achieve peak data rates, which is challenging because frame arrivals are \textit{unpredictable}. To do so, we introduce the idea of partitioning the set of banks into disjoint \textit{bank interleaving groups}, defined to be groups of $\gamma$ \textit{consecutive} banks $\ell$ to $(\ell + \gamma - 1)$. For example, with $\gamma = 4$, $L=64$ banks are partitioned into $L/\gamma=16$ consecutive groups of $\gamma=4$ banks. The first group comprises banks $1$ to $4$, the second banks $5$ to $8$, etc.
Then, given $T$ parallel HBM channels and $\gamma$ banks, \name first writes $1/\gamma\th$ of the frame into bank $\ell$ across all $T$ channels, then $1/\gamma\th$ into bank $\ell+1$, and so on into all banks in the group in a perfectly \textit{staggered} bank interleaving manner. $\gamma$ is defined in a way that makes all these bank accesses independent of each other without breaking HBM memory access rules.

\Ts{(4)~No bookkeeping.} To avoid complex bookkeeping and ensure frame ordering, \name deterministically writes  consecutive frames for the same output to consecutive bank groups, regardless of frame arrival patterns, \ie 
the $n$-th frame for a given output
is written into bank interleaving group $h = \para{n \mod (L/\gamma)}$.
On the output side, frames are read in the same sequence, thus keeping the frame order.

\Ts{(5)~Cyclical output reads.} \name reads the HBM cyclically for each of the $N$ outputs. It then proceeds at the output side symmetrically to the input side, reading frames from the HBM as slices across $N$ memory modules, cutting them into sliced batches, sending all batch slices to their output using an $N\times N$ cyclical crossbar, and finally unpacking batches into variable-size packets.

\Ts{(6)~Performance guarantees.}
As detailed in \cref{sec:guarantees}, \name guarantees \textit{100\% throughput.} 
In addition, with a small speedup, an \HBM with \name can \textit{mimic} an ideal OQ shared-memory switch, \ie 
given the same input sequence to the \HBM and to an ideal switch, any packet departs the \HBM within a finite delay after its departure from the ideal one~\cite{attiya2010}.

\begin{figure}
    \centering
    \includegraphics[width=\linewidth]{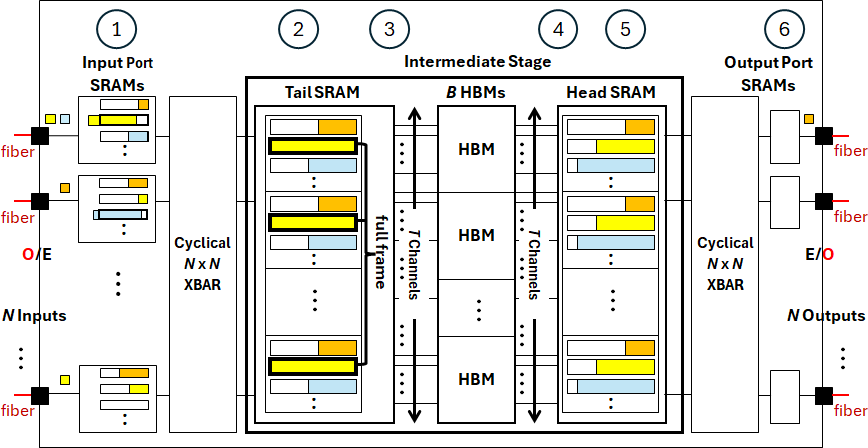}
    \caption{The architecture of each \HBM.
    }
    \label{fig:HBM-N}
\ReduceVSpace %
\end{figure}

\subsection{Architecture and \name Algorithm} 
\label{sec:HBM}

In this section, we detail the internal architecture of the \HBM (\Cref{fig:HBM-N}) and outline how \name handles traffic (\Cref{fig:intuition-N}).

\T{\ding{192} Input port SRAMs.} At each input port, %
incoming optical data from the waveguides is first converted into electronic signals via O/E conversion. Then, a processing chiplet determines the \HBM output for incoming variable-length packets, queues them in SRAM by output destination using $N = 16$ per-output queues, then packs them into \batches of $\k=4\unit{KB}$. Packets may straddle two batches. Once formed, batches for all outputs are (logically) stored in a later FIFO queue. The head-of-line batch is sliced into slices of $\k/N=256$ bytes. %
Then a cyclical crossbar sends each slice to a  different SRAM module in the tail SRAM, always starting from the first SRAM module.

\Ts{Example.} In the first input port of \cref{fig:HBM-N}, 
yellow packets 
for output 2 are queued in the second queue. They have already formed a \batch that is being gradually sent. Additional packets are arriving and starting to form a second batch. 

\Ts{Batch size.} Each input port SRAM must support the writing and reading of packets at the  data rate of $P=\alpha \cdot W \cdot R = 4 \cdot 16 \cdot 40\mbox{\,Gb/s} = 2.56\unit{Tb/s}$,
for a total of $2 P = 5.12\unit{Tb/s}$.
Assuming a $2.5\unit{GHz}$ clock, each input port SRAM can deliver
$2.5\unit{Gb/s}$ 
per bit of SRAM interface. 
Thus, the input port SRAM should provide a $5120/2.5 = 2,048$-bit wide interface. 
The batch size is exactly $N$ times this interface width, $k = 16 \cdot 2,048\unit{bits} =4\unit{KB}$ so that the batch slices can be uniformly spread across the $N$ SRAM modules at the tail SRAM.

\T{\ding{193} Tail SRAM.} The tail SRAM stores the batches, aggregates them into frames, then supports the concurrent writing of the $T$ segments in a frame over the $T = 128$ channels available across the ultra-wide interface of the $B = 4$ HBMs.

\Ts{(i)~Cyclical crossbar.} 
An $N \times N$ cyclical crossbar on the input side connects the $N$ input ports to the $N$ SRAM modules in the tail SRAM, using a cyclic rotation (no scheduling). The crossbar facilitates the input ports in striping frames across the $N$ tail SRAM modules.
Each crossbar port operates at an input data rate of $P = 2.56\unit{ Tb/s}$ via a $2,048$-bit interface.
As the crossbar implements a cyclic rotation, it can just be implemented by simple 1-dimensional multiplexors at the crossbar outputs with cyclic selections.
Alternatively, the crossbar could be replaced with an $N \times N$ {\em mesh} using spatial-division-multiplexing.
\ie the $2,048$ bits available at each input can be split into $N$ sets of $2,048/N = %
128$ wires, each set going to one of $N$ outputs for parallel transfers rather than cyclic connections, which may simplify the overall logic.

\begin{figure}
    \centering
    \includegraphics[width=\linewidth]{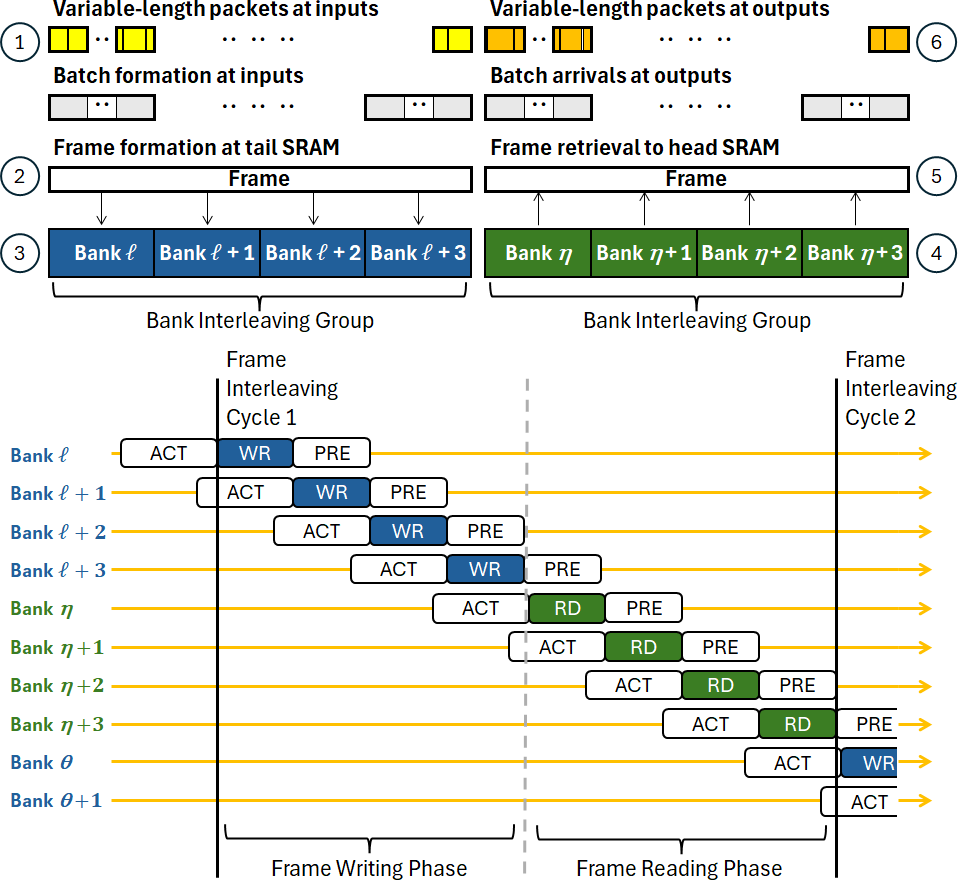}

    \caption{Intuition for the \HBM operation.}
    \label{fig:intuition-N}
\ReduceVSpace %
\end{figure}

\Ts{(ii)~SRAM modules.}
The tail SRAM is organized into $N$ separate SRAM modules. Each module is divided into $N$ logical areas corresponding to $N$ per-output queues. Each batch reaches all modules in a staggered way, such that each batch slice is stored in a per-output queue in its module. When the queue size of a module reaches $K/k=128$ batch slices, it forms a new frame slice, all modules doing so for the same frame in a staggered way. 
Frame slices for all outputs are then stored in a shared logical FIFO queue.

\Ts{(iii)~Memory width.} Each tail SRAM module is $2,048$-bit wide, operating at 2.5 GHz, as for each input port SRAM module.
Each HBM channel is $64$-bit wide, for a total of $512$ bits per group of $T/N = 128/16 = 8$ channels.
As HBM4 operates at 10 Gb/s per bit, \ie $4\times$ the data rate per bit in comparison with the 2.5 GHz SRAM clock, the $2,048$-bit SRAM interface can be readily serialized 4-to-1 to the $512$-bit HBM interface via the HBM controller. 
The group of $4$ HBMs provides $T = 4 \cdot 32 = 128$ channels across an ultra-wide interface, so $4 \cdot 2,048 = 8,192\unit{ bits } = 1,024\unit{ bytes}$ could be 
written into or read from the HBM group 
in parallel. 

\T{\ding{194} Writing into the HBMs.}
The blue elements in \cref{fig:intuition-N} illustrate how \name writes a frame into the HBMs. To exploit the ultra-wide HBM interface with $T$ channels, \name first writes a segment (blue WR) of size $S$ at bank $\ell$ into each of the $T$ channels in parallel. Using staggered bank interleaving, it does so $\gamma=4$ times for each frame, into  consecutive banks $\ell,\ell+1,\ell+2,\ell+3$. Thus, the frame size is $K = \gamma \cdot T \cdot S\unit{bytes}$. For each bank write, \name  in parallel closes the previous bank (precharge \textit{PRE}) and opens the next one (activation \textit{ACT}). 

Our reference design uses $S = 1\unit{KB}$, which is the smallest integer multiple of the HBM4 burst-length that satisfies the HBM four-activation window constraint~\cite{jedec} with our bank interleaving schedule, \ie such that at most four banks are activated at once, while also being a unit fraction of a row length. 
We also choose $\gamma$ so that (i)~the precharge (closing) of the first bank in one group completes before the activation of the first bank in the next, to ensure that we can continue the perfect staggered bank interleaving seamlessly from one group to the next, \ie each group is independent from the previous one; and (ii)~$\gamma$ should also satisfy a four-activation window limit that allows at most four concurrent bank activations, to prevent the memory from drawing too much instantaneous current. $\gamma = 4$ satisfies both conditions.
As a result, the frame size is $K = \gamma \cdot T \cdot S = 4 \cdot 128 \cdot 1\unit{KB} = 512\unit{KB}$.

\T{HBM memory organization.} 
The HBM memory is divided into per-output regions. 
This region allocation could be static, or dynamic with large per-output pages. Each region is then subdivided hierarchically: (1)~first into rows; (2)~then into segment-size sub-rows; and (3)~finally into banks. We write and read data in this exact FIFO order. For example, in \cref{fig:intuition-N}, after writing a blue segment into some bank $\ell$, we write the next one in the next bank $\ell+1$, across all $T$ channels.
With a static region allocation, the head, tail, and number of entries of the FIFO can simply be tracked with counters. With dynamic allocation using large per-output pages, a small extra amount of SRAM would suffice to track pointers to these large pages. 

\T{\ding{195} Reading from the HBMs.} \name cyclically reads a frame for each output.
\cref{fig:intuition-N} shows how it uses staggered bank interleaving, as for writes. It accesses the first bank of the next output, reads a segment  (green RD) across the $T$ channels, then does the same for all $\gamma=4$ consecutive banks, collecting a full frame for this output.

\T{\ding{196} Head SRAM.} 
Like the tail SRAM, the head SRAM is organized into $N$ SRAM modules.
Each head SRAM module obtains a full frame slice and cuts it into batch slices that are first stored in a per-output queue, then sent to their  output.
(See \Cref{sec:add_architecture} for more details.)

\T{\ding{197} Output ports.} 
At each output port, the received batches are cut back into variable-length packets, which are then converted back into optical signals via E/O conversion, before going out a waveguide. As the waveguides will be going out on an output fiber ribbon, which will go to another router, we have the freedom to carry the signals on any of the $\alpha = 4$ fibers and any of the $W = 16$ output WDM channels. As in ECMP or dynamic link aggregation (LAG)(\cref{sec:analysis}), we hash the packets across the available waveguides and wavelengths using their flow 5-tuples.

\section{Design Analysis}
\label{sec:analysis}

\T{Router buffer sizing.} The number of HBM memories should be sized to provide enough (i)~\textit{bandwidth} and (ii)~\textit{buffering}.
We found above that to provide the necessary memory bandwidth, we need to group $B = 4$ HBM4 stacks together per \HBM. Using 64\,GB per HBM4 
stack~\cite{HBM-wiki}, 
this translates into a large total amount of $H \cdot B\cdot 64 = 16 \cdot 4 \cdot 64 = 4.096$\,TB of buffering.
To put it into proportion, it means that the router can store up to $(H \cdot B\cdot 64) \cdot 8 /(N \cdot F \cdot W \cdot R) = 4.096\cdot 8/655.36 \approx 51.2$\,ms of buffering. 
This is in line with the old Van Jacobson rule of thumb that a buffer should hold one bandwidth-delay product of packets~\cite{jacobson1990modified}, and much more than the Stanford model of buffer sizing~\cite{Appenzeller04,redux}. It is also much more than a ``core router buffering in the range of 5-10 msec,'' as recommended in a Cisco white paper~\cite{cisco_white_paper}. Cisco's 400-Gbps linecards offer up to 18~ms (Q100 linecard) or 13~ms of buffering (Q200 linecard), \eg $5\unit{ms}$ for Cisco's 8201-32FH~\cite{cisco,queuepilot}.

\T{SRAM sizing.} %
The following theorem provides an upper bound for the total SRAM size (all proofs in the paper are in \cref{sec:proofs}).
    \SmallReduceVSpace
    \begin{theorem}[SRAM size]\label{thm:bounded}
    The total SRAM size is at most $\frac{3}{2}\para{N+1}\K+\para{N+2}N\cdot \k -N^2$. 
    \end{theorem}

We find that the total needed SRAM size is $14.5\unit{MB}$,
which can be easily implemented today in SRAM. 
It is a small cost we pay for assembling the large frames that are needed to fully exploit the large HBM parallelism deterministically while accessing memory at peak data rates. Note that if we were to adopt an alternative statistical approach, akin to packet spraying in datacenters, we would not need to pay this cost of forming frames, but would need to pay an alternative memory cost for the packet reordering buffer, which seems to be an order of magnitude higher depending on the acceptable reordering rate~\cite{yang2017simple,usubutun2023switches,rottenstreich2012switch}.

\T{Power estimate.} Each of our $H=16$ \HBMs handles $41\unit{Tb/s}$ of incoming traffic. This is less than the Broadcom Tomahawk 5 BCM878900 switch chip, which can handle $51.2\unit{Tb/s}$ incoming traffic~\cite{TH5}
with a power dissipation of $500$\,W~\cite{TH5-500}.
The Broadcom chip can handle all the packet processing and SRAM buffering. Thus, as a first approximation, the packet processing and SRAM buffering for each \HBM should consume at most $500 \cdot (41/51.2) = 400$\,W.
In addition, each HBM4 stack should consume about $75$\,W~\cite{HBM-nextgen}, so $B=4$ stacks consume about $300$\,W. 
Finally, commercially available silicon photonics can perform OEO conversion at about 1.15 pJ/bit~\cite{GF-fotonix, Moazeni2017,Levy2023CICC,Li2024OJSSCS,Raj2023ISSCC,Hsueh2025}.
At $81.92\unit{Tb/s}$ of I/O per \HBM, the power required for OEO conversion for each \HBM is about $94$\,W.
Overall, each \HBM should consume about $400+300+94 = 794$\,W, yielding about $12.7$\,kW for all $H = 16$ \HBMs. %
This can be compared to the Cerebras WSE-3 wafer-scale processor~\cite{Kennedy24}, which is commercially deployed and consumes $23$\,kW of power. It uses a state-of-the-art liquid and air cooling system~\cite{Mujtaba21, Kennedy24, Kundu25} to dissipate the heat. %
The same system could be used to cool our internet router, which consumes just above half the power.

\T{Area estimate.} The  Broadcom die size is estimated at $800\unit{mm}^2$~\cite{TH5}.
Assuming conservatively $800\unit{mm}^2$ for the processing chiplet and $4 \cdot (11\unit{mm} \cdot 11 \unit{mm}) = 484\unit{mm}^2$ for the HBM stacks, we get about $1,284\unit{mm}^2$ per \HBM. For the 16 \HBMs, we obtain $16 \cdot 1284\unit{mm}^2 = 20,544\unit{mm}^2$, which is under 10\% %
of the $500\unit{mm} \cdot 500\unit{mm} = 250,000\unit{mm}^2$ of the surface area of a panel-scale substrate. Thus, the area should not be a bottleneck at this stage.

\T{Traffic matrix at \HBMs.} 
We assume a spatial split of the incoming fibers into the HBM switches. This may possibly lead to uneven traffic matrices (TMs) between HBM switches.
In practice, most incoming links would use either \textit{ECMP} or dynamic \textit{link aggregation (LAG)} bundling~\cite{ergun}. LAG is defined in the link aggregation control protocol (LACP) in the IEEE %
802.1AX standard~\cite{ieee-lacp}, bundling multiple physical fiber-optic connections into a single logical link. Both with ECMP and with LAG, traffic is load-balanced across the physical links using a hashing algorithm, \eg based on the flow 5-tuple~\cite{juniper-hashing,nokia-hashing}. 
Thus, all packets of a flow are sent on the same physical link to prevent reordring. If hashing achieves a good distribution, then a spatial split would leave a similar TM at each HBM switch, yielding good performance. This is formally proved in \cref{thm:global-full-throughput}~(\cref{sec:guarantees}). 

However, if the different fibers come from distinct sources, then the uneven TMs may cause packet loss. For example, assume that with $N=2$, $F=2$ and $W=1$, the $4 \times 2$ normalized traffic matrix (TM) looks like: 
\be 
TM =
 \begin{matrix}
\substack{\text{Input 1: Fiber 1} \\ \text{Input 1: Fiber 2}} \\
\\
\substack{\text{Input 2: Fiber 1} \\ \text{Input 2: Fiber 2}}
\end{matrix}
 \begin{pmatrix}
 \underline{\textit{0.6}} & \underline{\textit{0.3}}  \\
  0.3 & 0.6   \\
  \hline 
 \underline{\textit{0.5}} & \underline{\textit{0.4}}  \\
  0.3 & 0.15 \\
 \end{pmatrix},
\ee
where each row represents a fiber, the first group of two rows represents an input, and each column represents an output. The sum on each column (output) is less than 2, therefore it can be serviced by the 2 output fibers and the TM is admissible. 
 After the spatial split, the first HBM switch may get the first fiber (first row) of each input (shown in underscored italics):
\be
TM_\text{HBM} =
 \begin{pmatrix}
  0.6 & 0.3  \\
  \hline 
 0.5 & 0.4  \\
 \end{pmatrix}.
\ee
The sum on the first column is 1.1, exceeding the HBM switch output capacity of 1. Thus, even if the HBM switch is optimal, it will drop packets at a rate of 0.1
out of its total arrival rate of 1.8, yielding a loss rate of 5.6\%. We can see how this spatial split is sensitive to large TM elements that are in different fibers and have correlated outputs. %

Evaluations show that expected packet losses are much lower than in this simple toy example, often under 0.1\%~(\cref{sec:eval}).
Also, a slightly smaller switching capacity will simply lead TCP flows to reduce their rate and adapt to it, thus naturally decreasing the loss rate.

\T{Latency and bypass.} Waiting to fill frames for scheduling can increase latency. Instead, when there are no full frames, we can use frame padding~\cite{padded-frames,thesis} to decrease latency. 
A \textit{bypass} mechanism can further reduce latency: If the tail SRAM sees that there is nothing stored in the HBM for output $k$, then when it is time for the head SRAM to cyclically get a frame for output $k$, the tail SRAM can bypass the HBM and send the head-of-line (potentially padded) frame  directly to the head SRAM.

\T{Frame interleaving cycle.} 
The transitions between the frame write/read phases total about $2\%$ of the cycle duration~\cite{jedec}.
Since reads and writes are mandatory in a switch, 
we consider 
these 
transition delays
as part of our baseline 100\% throughput.
Also, HBM4~\cite{jedec} provides independent single-bank refresh operations that can be readily hidden 
without affecting the frame interleaving cycle time. 
See \Cref{sec:add_architecture} for more details.

\T{Clock.}
The electronics in the proposed design are confined within each HBM switch, with photonic connections among the $H$ HBM switches. Clock synchronization of chiplets with multiple HBM stacks at the scale of our HBM switch has already been demonstrated in commercial processors~\cite{nvidia-blackwell,nvidia-blackwell-ultra}. Clock distribution and power delivery to the $H$ HBM switches occur in the redistribution layer (RDL) of the panel substrate.

\T{Reliability.} The parallel-switch architecture relies on many parallel and independent smaller  switches. Similarly to sharding and network slicing, it enables an incremental switch management, \eg by gradually introducing new OS versions in the smaller switches. It also reduces the failure blast-radius 
of any fiber or smaller switch.

\T{Packing packets.} When packing packets into frames, a natural question arises whether we want to allow a variable-length packet to span across two frames, \ie at the end of the first and the beginning of the next. It increases utilization, but then the two frames from the same input may not arrive at an output back-to-back, as there may be another frame from another input in-between. Thus, this may slightly increase memory needs and complexity. On the other hand, not packing packets only loses $\frac{1.5\unit{KB}}{128\unit{KB}}\approx 1\%$  of utilization.

\section{Performance Guarantees}\label{sec:guarantees}

We want to demonstrate that the \HBM ``\textit{behaves like}'' an ``\textit{ideal}'' switch. Formally, (i)~by \textit{ideal}, we mean a \textit{FIFO OQ shared-memory switch}, \ie a switch where all incoming packets arrive directly at the egress and are serviced in order; and (ii)~by \textit{behaves like}, we mean that the \HBM switch \textit{mimics} the OQ shared-memory switch, \ie there exists some fixed delay $D$ such that if both switches get the same packets at the same times, and some packet exits the OQ shared-memory switch at time $t$, then it will leave the \HBM before time $t+D$.

The \HBM needs to collect frames. Therefore, it does \textit{not} directly mimic a shared-memory switch. For instance, if a single small packet is the only incoming packet, it will immediately leave a shared-memory switch, while it may get stuck forever at an unmodified \HBM, waiting to form a frame. This may be the curse of modern wide memories that intrinsically need frames to achieve high rates. 
However, we can still prove strong properties of the \HBM. We show that for normal traffic, it obtains \textit{100\% throughput}, like the shared-memory switch. We also show that if all packets are multiples of the frame length, then it does mimic a shared-memory switch. Finally, we prove that the \HBM also mimics a shared-memory switch if it is allowed to have a \textit{small speedup}  to clear leftover packets using padded frames.

In all these results, we assume that both a shared-memory switch and the \HBM have infinite HBM buffer sizes, thus decoupling the switching from the specific memory partition policy that is orthogonal to this paper~\cite{abm, reverie}. %

\SmallReduceVSpace
\begin{theorem}[Mimicking with frames only]\label{thm:mimic-frames}
Given only frame arrivals, \ie arrivals of frame-sized bursts with a common output for each burst, the \HBM mimics a FIFO OQ shared-memory switch.
\end{theorem}
\SmallReduceVSpace

Define an arrival traffic as \textit{admissible} whenever a shared-memory switch can service it with a bounded queue size, and define a switch as having \textit{100\% throughput} whenever it can service any admissible arrival traffic with a bounded queue size. 
We obtain: %
\SmallReduceVSpace
\begin{theorem}[100\% throughput for \HBM]\label{thm:full-throughput}
The \HBM achieves the same 100\% throughput as an OQ shared-memory switch (even without speedup).
\end{theorem}
\SmallReduceVSpace

We can further generalize this result from the smaller \HBM to the large \arch architecture if we 
model the ECMP or LAG hashing (\cref{sec:analysis}) among the fibers of each input 
as achieving an ideal distribution. 
\SmallReduceVSpace
\begin{theorem}[100\% throughput for \arch]\label{thm:global-full-throughput}
If at each input, all fibers have the same traffic, then the \arch architecture with a set of smaller \HBMs guarantees 100\% throughput. %
\end{theorem}
\SmallReduceVSpace

We finally prove that the \HBM can mimic a shared-memory switch given any small speedup $1+\epsilon>1$ for its memory components.  
\SmallReduceVSpace
\begin{theorem}[Mimicking with speedup]\label{thm:mimic-speedup}
An \HBM with a small speedup can mimic the behavior of a FIFO OQ shared-memory switch.
\end{theorem}
\SmallReduceVSpace

\section{Evaluations}
\label{sec:eval}

\T{Overview.} We introduced a new coarse distribution technique based on a simple spatial split of the optical fibers, followed by a fixed pseudo-random assignment to the \HBMs. We saw that if incoming links come from a router that implements ECMP or LAG-based flow hashing   between fibers, different fibers will generally have a similar distribution and the spatial split is not expected to significantly impact throughput~(\cref{sec:analysis}).
However, without ECMP or LAG, the physical fiber split among smaller \HBMs operating at a reduced capacity may result in poor load balancing of the traffic, leading to packet losses. The goal of this section is to evaluate the added loss rate resulting from the poor mixing of this fiber split, assuming that an ideal load-balancing does not generate any loss. 

The evaluations reveal the following key results when using fiber split: 
\begin{itemize}
    \item Using a backbone network workload reveals that there is \textit{zero packet loss} in both Abilene and GEANT networks when implementing the original router traffic matrices, and 
    less than 0.013\% when fitting router sizes to the large size used in this paper~(\cref{sec:network}).
    \item Using a real backbone CAIDA link trace with a synthetic wavelength allocation 
    shows a loss rate 
    of under 0.008\%~(\cref{sec:caida}).
    \item Considering the more challenging cross-DC AI training workloads,  across all evaluated load-balancing schemes, our router-in-a-package architecture maintains under 0.35\% packet loss rates~(\cref{sec:AI}).    
    \item Finally, applying a sensitivity analysis to synthetic WAN traffic shows a loss rate under 0.05\% for both uniform and non-uniform cases (\cref{sec:sim}).
\end{itemize}
Thus, evaluations consistently suggest that the impact of coarse load-balancing should stay minimal.

\subsection{Backbone network workload} \label{sec:network}

\T{Setup.} The primary objective of the proposed router architecture is the implementation of Internet WAN backbone routers. Accordingly, we examine the impact of splitting backbone traffic among the smaller \HBMs. We use SNDlib~\cite{SNDlib} to obtain 
two WAN topologies and network demand-matrix traces: the Abilene and GEANT networks. We extract $30$ consecutive time samples from each trace, corresponding to 30 consecutive intervals of $5\unit{minutes}$ each for Abilene and $15\unit{minutes}$ for GEANT.
To obtain router-level traffic distributions, for each sample, we apply the POP~\cite{narayanan2021solving}  WAN traffic engineering (TE) approach, solving a multi-commodity flow linear program (LP) using Gurobi solver~\cite{gurobi}. 
This LP maps the end-to-end demands into paths over the given topology while respecting link capacities, yielding a TM for each router in the network. 

\T{Metrics.} Given a set of TMs over all routers in the network, we evaluate the loss rate using two metrics:  
\bl
    \item \textit{Average router loss rate,} \ie the average of the individual loss rates at each router.
    \item \textit{Network loss rate,} \ie the total packet loss rate across the entire network divided by the total incoming traffic rate.
\el

\begin{figure}
\centering
\begin{subfigure}{0.48\linewidth}
  \includegraphics[width=\textwidth]{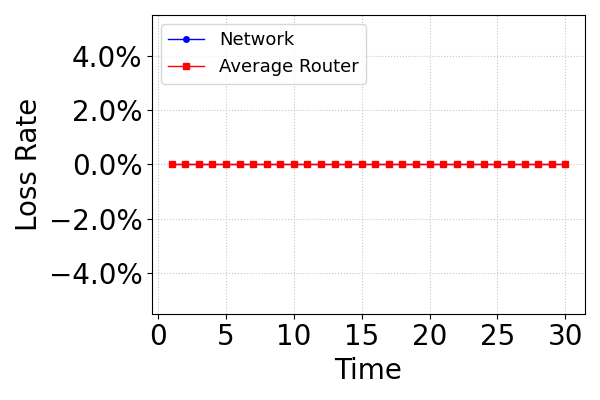}
  \caption{Abilene}
  \label{fig:Ab_no_re}
\end{subfigure}
\hfill
\begin{subfigure}{0.48\linewidth}
  \includegraphics[width=\textwidth]{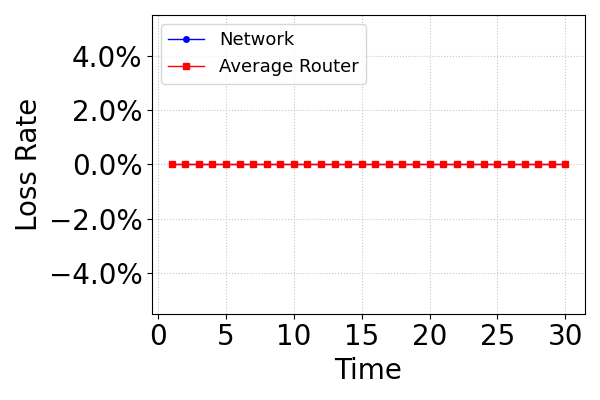}
  \caption{GEANT}
  \label{fig:GE_no_re}
\end{subfigure}
\caption{Packet loss rate (or lack thereof)
with Abilene and GEANT workloads and original router sizes.
}
\label{fig:Backbone_no_resize}
\ReduceVSpace %
\end{figure}

\begin{figure*}[!t] %
\centering
\begin{subfigure}{0.24\linewidth}
  \includegraphics[width=\textwidth]{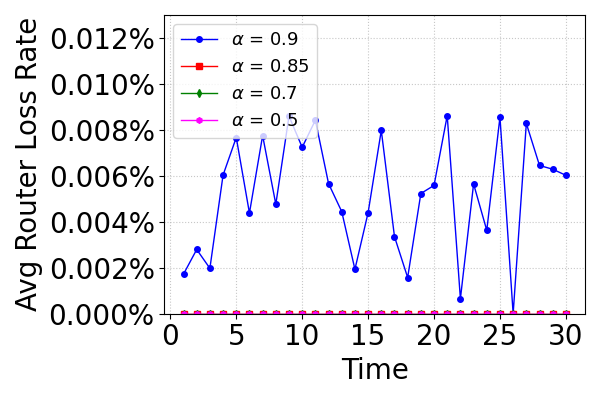}
  \caption{Abilene (router average)}
  \label{fig:Ab_avg}
\end{subfigure}
\hfill
\begin{subfigure}{0.24\linewidth}
  \includegraphics[width=\textwidth]{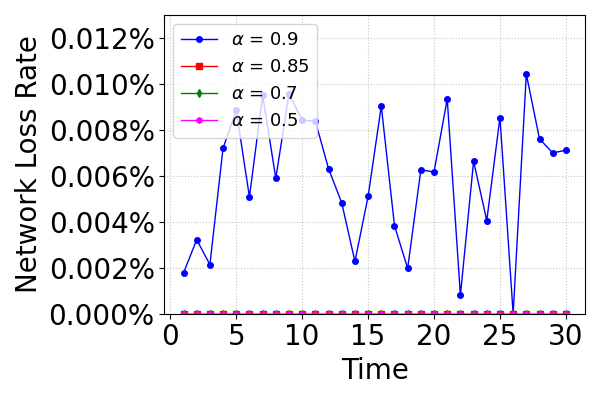}
  \caption{Abilene (network average)}
  \label{fig:Ab_net}
\end{subfigure}
\hfill
\begin{subfigure}{0.24\linewidth}
  \includegraphics[width=\textwidth]{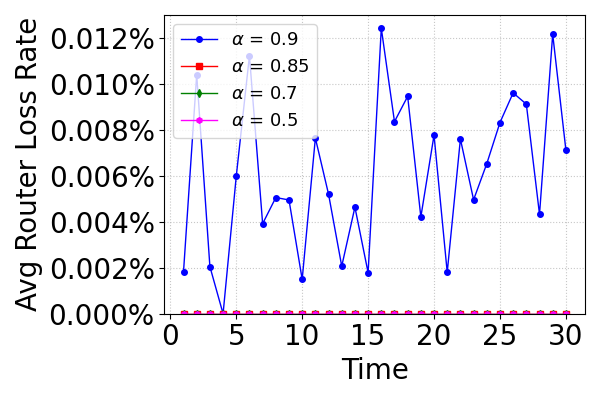}
  \caption{GEANT (router average)}
  \label{fig:GE_avg}
\end{subfigure}
\hfill
\begin{subfigure}{0.24\linewidth}
  \includegraphics[width=\textwidth]{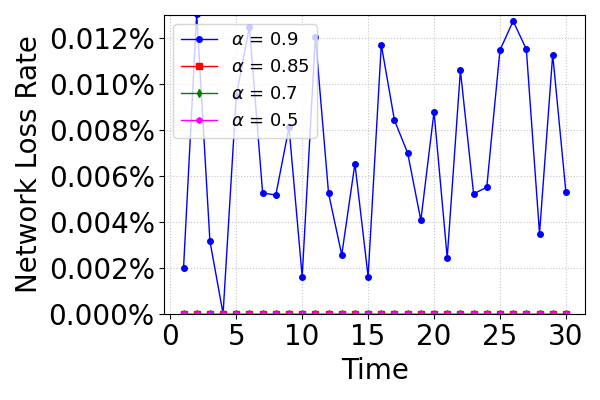}
  \caption{GEANT (network average)}
  \label{fig:GE_net}
\end{subfigure}
\hfill
\caption{Packet loss rate 
using TM resizing to large routers, with different $\alpha$ values.
}
\label{fig:Backbone_resize}
\ReduceVSpace %
\end{figure*}

\T{Original router sizes.} 
We start with the original router sizes.
For each router we assume that the number of \HBMs $H$ equals the number of ports $N$, and that each network flow uses a different fiber, thus using a random \HBM at each router. We scale the rates so that the most loaded input or output of the most loaded router is loaded up to $1$, so all router TMs are admissible.

\Cref{fig:Backbone_no_resize} shows how the resulting loss rate is \textit{zero} for both WANs in all time samples. The pseudo-random fiber split is sufficient to obtain 100\% throughput in this case.

\T{TM resizing.} We convert the small router TMs into large ones with an $\para{N\cdot F \cdot W} \times N = 16,384 \times 16$ router. This conversion involves three scaling challenges: scaling \bl \item the number of inputs, \item the number of outputs, and \item $R$.\el 
We expand each router TM %
to our desired dimension, while preserving the overall traffic and distribution characteristics. Recursively, we split rows and columns. Each row split doubles the number of rows. For each row, we divide its traffic across the two new associated rows by sending a fraction $\alpha \in [0.5, 1]$ of the load to the top sub-row, and $1-\alpha$ to the bottom. This splitting process is repeated from the original number of rows $n$ until the expanded TM reaches size $n \cdot 2^i$, where $2^i$ is the largest factor such that $n \cdot 2^i \leq 16,384$. The same procedure is applied to the columns, ensuring that both sources and destinations are expanded consistently. We scale the rates to obtain a maximum input or output load of 1.

\Cref{fig:Backbone_resize} illustrates the resulting average router loss rate and the network loss rate at Abilene and at GEANT, for different values of $\alpha$ and different time samples. In all cases, the loss rate is below 0.013\%, with loss occurring only when $\alpha = 0.9$.

\subsection{CAIDA traces}
\label{sec:caida}

\T{Setup.} We further analyze the potential loss rate due to the spatial split using real backbone traffic. We use a dataset from CAIDA~\cite{caida2025passive100g} that provides an anonymized packet header captures from a $100 \unit{Gbps}$ backbone link between Los Angeles and Dallas. 
We analyze $60$ consecutive intervals of $100$ms (in total $6$ seconds) from each link direction of a capture taken in August 2025. For each flow defined as a tuple of source–destination IP pairs, we hash the $16$-bit prefix of the source IP address to select one of the $\para{N\cdot F \cdot W} = 16,384$ wavelengths, and hash the destination IP to select one of the $N=16$ output ports. By design, each wavelength is associated with one of the $16$ \HBMs. We use the MD5 cryptographic hash function, and scale the rates so that the most loaded input or output of the most loaded router is exactly loaded up to 1.

\T{Results.}
\Cref{fig:CAIDA_loss_rate} depicts the router loss rate over 60 consecutive measurement intervals of $100\unit{ms}$, with \Cref{fig:CAIDA_dirA} and \Cref{fig:CAIDA_dirB} representing the two opposing link directions. The results show minimal packet loss rates, consistently maintained below 0.008\% (\ie 1 packet loss every 12,500 packets) in both directions, validating the effectiveness of the split-fiber design. The temporal fluctuations observed in both subfigures correspond to dynamic changes in the trace throughout the measurement period.

\begin{figure}[!t] %
\centering
\begin{subfigure}{0.48\linewidth}
\includegraphics[width=\textwidth]{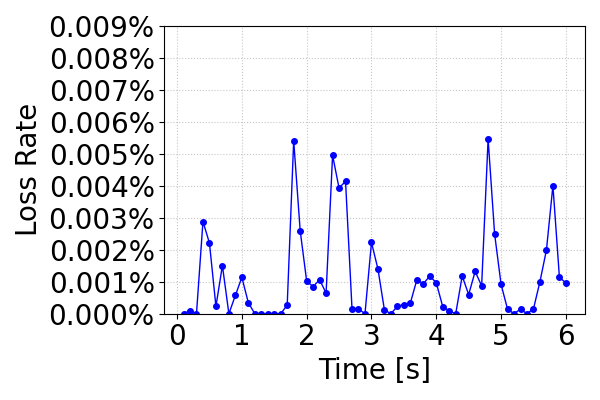}
  \caption{Direction 1}
  \label{fig:CAIDA_dirA}
\end{subfigure}
\hfill
\begin{subfigure}{0.48\linewidth}
  \includegraphics[width=\textwidth]{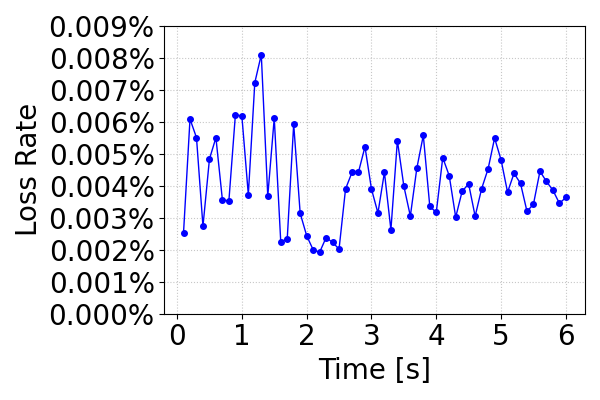}
  \caption{Direction 2}
  \label{fig:CAIDA_dirB}
\end{subfigure}
\caption{Router loss rates over two link directions in a CAIDA trace. 
}
\label{fig:CAIDA_loss_rate}
\ReduceVSpace %
\end{figure}

\subsection{Synthetic cross-datacenter AI workload} \label{sec:AI}

\begin{figure*}[!t] %
\centering
\begin{subfigure}{0.24\linewidth}
\includegraphics[width=\textwidth]{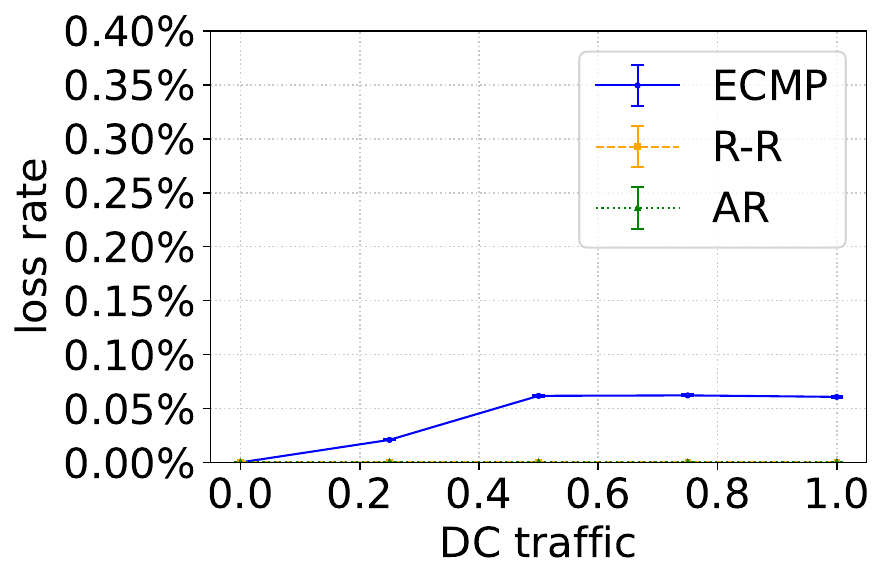}
  \caption{DC traffic}
  \label{fig:DC:01}
\end{subfigure}
\hfill
\begin{subfigure}{0.24\linewidth}
  \includegraphics[width=\textwidth]{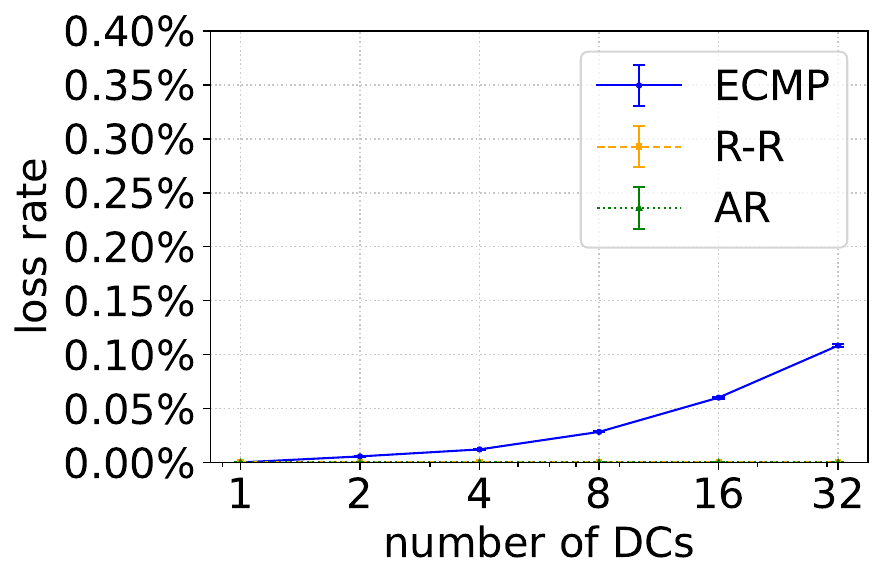}
  \caption{Number of DCs}
  \label{fig:DC:02}
\end{subfigure}
\hfill
\begin{subfigure}{0.24\linewidth}
  \includegraphics[width=\textwidth]{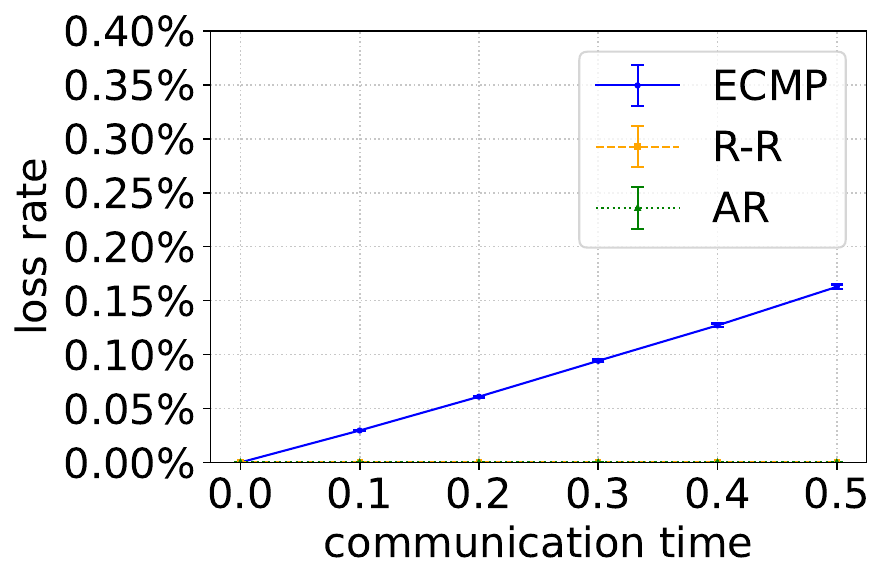}
  \caption{Communication time}
  \label{fig:DC:03}
\end{subfigure}
\hfill
\begin{subfigure}{0.24\linewidth}
  \includegraphics[width=\textwidth]{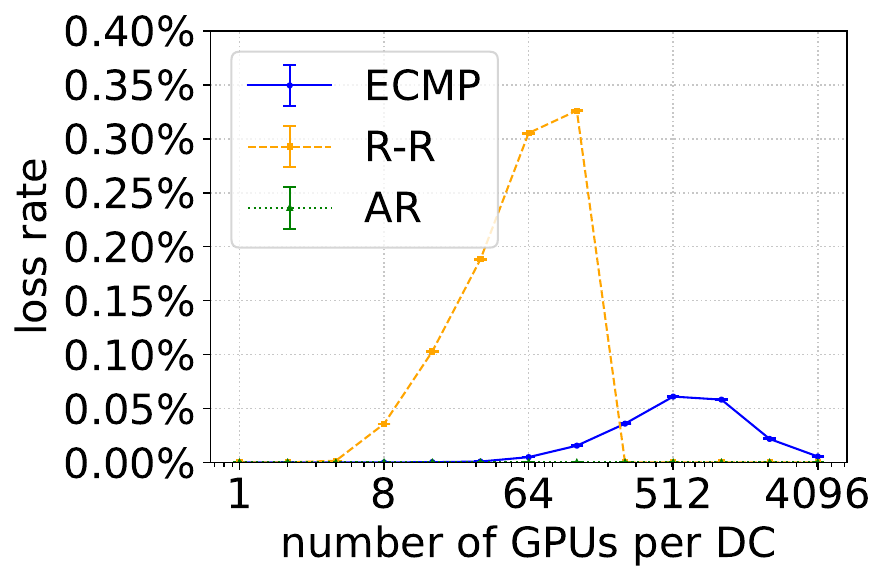}
  \caption{Number of GPUs per DC}
  \label{fig:DC:04}
\end{subfigure}
\caption{Impact of cross-traffic from AI training datacenters on the packet loss rate in a WAN router, given three load-balancing techniques across the DC egress wavelengths. Sensitivity to different parameters:
(a)~the DC ingress/egress traffic rate as a fraction of its link capacity;
(b)~the number of DCs in the PP cluster;
(c)~the communication time fraction out of the total time; and
(d)~the number of GPUs. %
}
\label{fig:DC}
\ReduceVSpace %
\end{figure*}

\T{Setup.} We want to study the impact of the emerging \textit{cross-datacenter (DC)} traffic during AI training across several DCs~\cite{chen2025crosspipe,SDR}. We consider a single cluster of $n$ DCs, each containing $m$ GPUs. Each DC is located at a distinct random router port. It can send and receive traffic at a fraction $\alpha$ of the port capacity $F \cdot W \cdot R %
\approx {40\unit{Tb/s}}$. 
The DC cluster applies cross-DC pipeline parallelism (PP)~\cite{chen2025crosspipe,cocoDC}: The model layers are divided into $n$ stages, and communication occurs at the stage boundaries, as activations and gradients are sent and received. Communication occurs a fraction $\beta$ of the total time. When it occurs, it is either a forward pass or a backward pass with equal probabilities, as in the interleaved 1F1B schedule~\cite{1F1B,chen2025crosspipe}. We assume a maximum allowed wavelength load of 0.95. If a DC sends too much traffic to a wavelength, we downscale all flows proportionally to satisfy the allowed load on this wavelength. Finally, cross-DC traffic has strict priority over the remaining WAN traffic. After computing the DC traffic, we generate synthetic WAN traffic and add it, up to the allowed load (as detailed in \cref{sec:gen}). 
The low-priority WAN traffic mostly fills the wavelengths that do not carry any DC traffic.
For each TM, we then use 100 random HBM assignments as previously to measure the loss. We always compute the average packet loss rate over all packets from all sources sent at all times, including those times without cross-DC communication.
We use the following four baseline values: $n=8$\,DCs; $m=512$\,GPUs per DC stage~\cite{rail}; $\alpha=1$ so a DC can use all of its available capacity (up to the allowed load); and a fraction $\beta=0.2$ of the time devoted to communication (\eg $\beta$ can vary from 1.9\% to 57.4\%, depending on the parameter-efficient fine-tuning method~\cite{alnaasan2024}). We then study the impact of varying each of these four values.

\T{Load-balancing algorithms.} We compare three algorithms that the DCs can use when load balancing egress traffic across all available wavelengths, extending current load-balancing techniques for internal DC traffic~\cite{2025loadbalancing}: 
\bl 
\item \textit{ECMP}, which hashes the GPU flows to pick a wavelength~\cite{SDR,meta}. We model the hashing as perfect with a uniform random choice.
\item A carefully-constructed schedule, extending similar ideas within DCs~\cite{ethereal,qian2024alibaba,samuel2017routing,meta,cohen2025routing,liu2024figret}. 
It is akin to spreading the flows in \textit{round-robin} order across the wavelengths, which we denote as \textit{R-R}. 
\item A per-packet \textit{adaptive-routing (AR)} algorithm, as used in NVIDIA's Spectrum-X within DCs~\cite{nvidia_AR_whitepaper2,ARbook, nvidia_AR_roce, nvidia_AR_whitepaper}. It is modeled as sending the same load on each wavelength. Note that AR is also subject to reordering.
\el

\T{Results.} 
\Cref{fig:DC} illustrates the loss rate as a function of (a)~normalized DC traffic bandwidth, (b)~the number of DCs, (c)~the normalized communication time, and (d)~the number of GPUs per DC. Loss rate is a bit higher, as DCs generate large flows that are more likely to make TMs non-admissible within HBM switches after the spatial load-balancing. Still, the loss rate mostly seems reasonable. It is typically higher for ECMP, which hashes flows and therefore does not load-balance them most efficiently. Interestingly, the worst case in \Cref{fig:sim:01,fig:sim:04}  is when half the maximum quantity is load balanced (specifically, half the traffic in \cref{fig:sim:01}, and half the number of $F\cdot W=1,024$ flows in \cref{fig:sim:04}). The complex intuition is that many wavelengths are then filled up to their maximum load, yet there is enough capacity left for the same DC output, so WAN traffic can fill other wavelengths from other inputs that are also spatially split to the same HBM switch. In addition, R-R behaves perfectly in most cases. However, in \cref{fig:sim:04}, when there are not enough flows to load-balance among the $F\cdot W=1,024$ wavelengths per DC, it maximizes the non-uniformity and therefore incurs more loss (up to 0.35\%). Finally, AR attains the best performance in all cases, since its per-packet load-balancing is optimal.

\subsection{Sensitivity analysis} \label{sec:sim}

\T{Setup.} We generate synthetic traffic matrices (TMs), as it enables us to model the sensitivity of the load-balancing performance to different parameters. The TMs should describe how the various ingress wavelengths traffic are divided among the different egress output ports. Thus, each TM has $N\cdot F\cdot W=16,384$ rows (input wavelengths) and $N=16$ columns (outputs). We proceed in several steps. 

(1)~ Since there are no benchmark router TMs, we generate the TMs using a balls-and-bins approach, by generating flows (\textit{balls}) and then assigning them to random (input wavelength, output) pairs (\textit{bins}), as long as the total load on each input wavelength and output does not exceed some load limits. Whenever this maximum load is at most 1, the TM is said to be admissible. We detail this synthetic TM building procedure in \cref{sec:gen}.

(2)~Once we obtain the TM, we generate 100 different random assignments to the HBM switches and compute the resulting TM at each HBM switch. We assume that each HBM switch has 100\% throughput, \ie it can service any admissible traffic matrix. However, once the rate destined for an HBM-switch output exceeds its capacity, the above-capacity traffic is dropped. The final loss rate is the ratio of the total drop rate by the total arrival rate. We assume a baseline Zipf exponent of $s=1$ (except when looking at a uniform TM), and a baseline number $H=16$ of HBM switches.

\begin{figure*}[!t] %
\centering
\begin{subfigure}{0.24\linewidth}
\includegraphics[width=\textwidth]{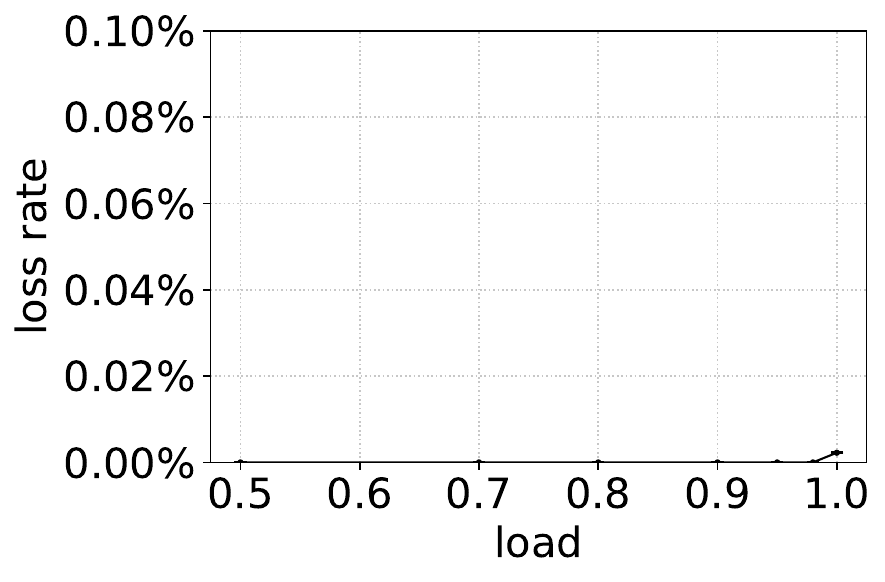}
  \caption{Load (uniform TM)}
  \label{fig:sim:01}
\end{subfigure}
\hfill
\begin{subfigure}{0.24\linewidth}
  \includegraphics[width=\textwidth]{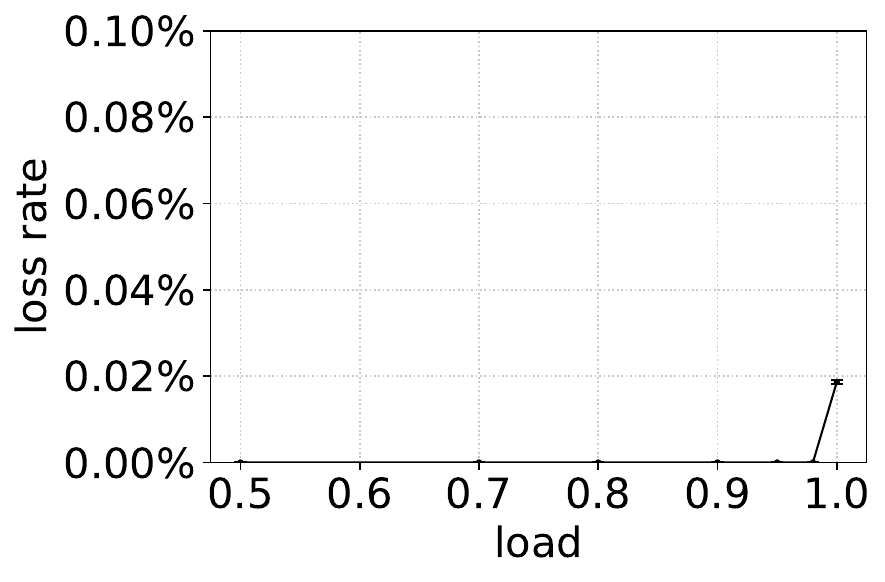}
  \caption{Load (non-uniform TM)}
  \label{fig:sim:02}
\end{subfigure}
\hfill
\begin{subfigure}{0.24\linewidth}
  \includegraphics[width=\textwidth]{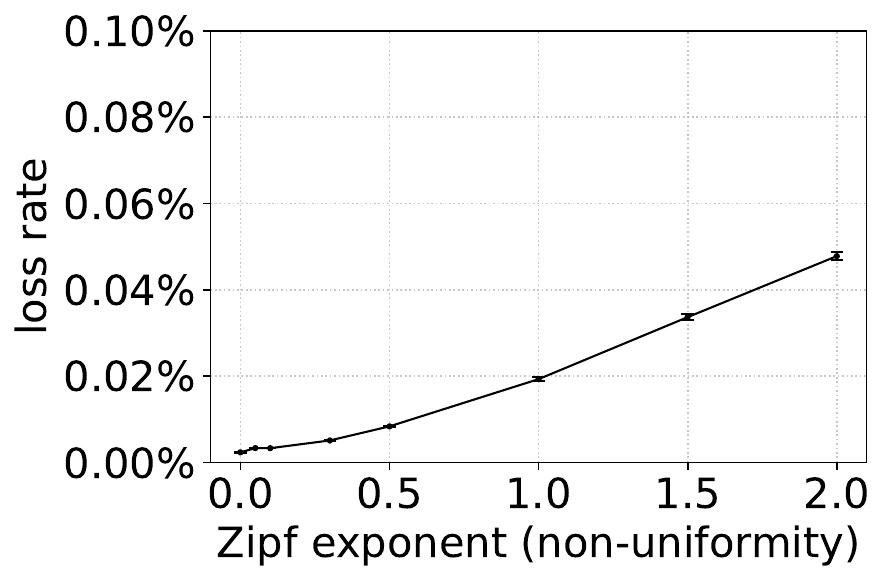}
  \caption{Non-uniformity}
  \label{fig:sim:03}
\end{subfigure}
\hfill
\begin{subfigure}{0.22\linewidth}
  \includegraphics[width=\textwidth]{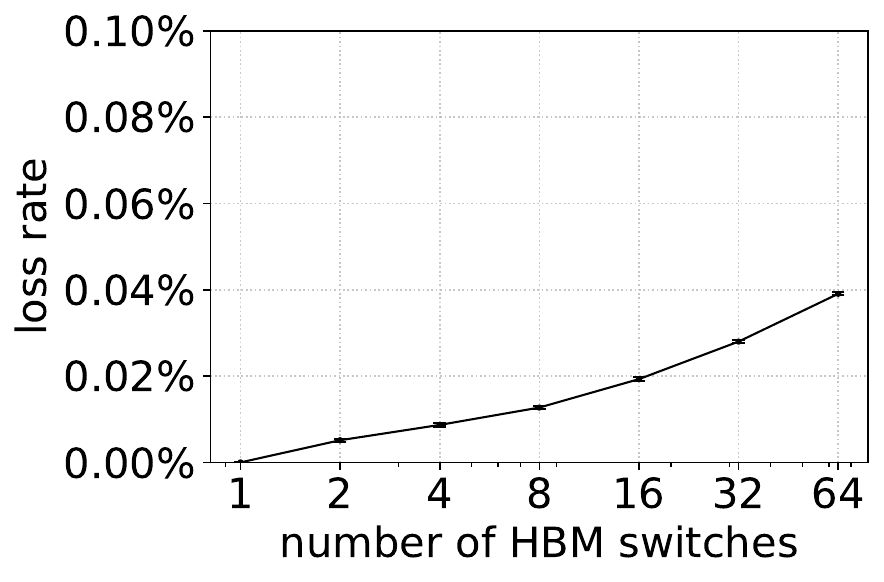}
  \caption{Number of switches}
  \label{fig:sim:04}
\end{subfigure}
\caption{Impact of various parameters on the packet loss rate in an internet router, assuming an admissible TM. (a) and (b)~plot the impact of load in a uniform and non-uniform TM, respectively. 
(c)~illustrates the sensitivity to the Zipf exponent, which reflects the non-uniformity of the TM. 
(d)~shows the impact of the number of HBM switches. %
}
\label{fig:synthetic}
\ReduceVSpace %
\end{figure*}

\T{Results.} %
\Cref{fig:synthetic} illustrates the packet loss rate as a function of (a-b)~the load (in a uniform and non-uniform TM), (c)~the Zipf exponent that represents the TM non-uniformity, and (d)~the number of HBM switches. In all cases, the loss rate is small and below 0.05\%. It slightly increases with non-uniformity and with the number of HBM switches.

\section{Impact on networking future}
\label{sec:discussion}
We now discuss the possible implications of our novel router design on the future networking landscape. 

\T{Capacity increase.} We have shown a potential increase in router capacity per given area by 1-2 orders of magnitude vs. current routers, given the emerging in-package photonics, heterogeneous integration, and HBM4 technologies; 
\eg a Cisco 8201-32FH of 1RU (rack unit) height and HBM %
can accept 32 lines of 400\,Gb/s for an aggregate 12.8\,Tb/s, over {50$\times$} 
less than the input bandwidth of our router, %
while occupying about the same space. %

\T{Router evolution.} Although our reference design is based on the latest HBM4 memories~\cite{micron,skhynix,skhynix-ee}, %
industry roadmaps show that future HBMs are expected to increase by a factor of 4$\times$ in memory capacity and bandwidth~\cite{HBM-nextgen}. %
Emerging memory technologies using monolithic 3D stackable DRAM (multi-DRAM layers per die) with direct stacking of memories and processing chiplets in the same 3D stack~\cite{Hsu25a, Hsu25b} are expected to increase the memory capacity and memory bandwidth of an HBM stack by a factor of 10$\times$ compared to HBM4. These expected improvements will enable us to realize our reference design with far fewer HBM stacks, translating into smaller footprints and power. They can also enable us to realize yet higher capacity routers in a single package.

\T{Buffer sizing.} Previous research assumed off-chip memory (DRAM) bandwidth is a key limiting factor in router performance, which prompted much research into minimizing
router buffer sizes so that faster routers could be designed with only on-chip memory without off-chip access necessity~\cite{redux}.
However, with the integration of many HBM stacks and packet processing chiplets within the same package, neither memory capacity nor memory bandwidth 
appears to be a barrier to faster routers, even when scaling to handle massive incoming traffic made possible by in-package photonics I/Os. 
This \textit{memory glut} means that router buffer sizing research should be revisited. This does not mean that smaller buffers should be automatically discarded. Smaller buffers may still provide better QoS, as in the Netflix experiment~\cite{netflix}, and they may reduce power consumption, but their impact on router capacity is smaller than previously assumed.

\T{Buffer management.} The assumption that ``buffer size is not keeping up with the increase in switch capacity''~\cite{abm} may no longer hold.
Thus, the memory glut may also impact buffer management and buffer-sharing algorithms~\cite{abm, reverie},
reducing the need for complex algorithms to address memory scarcity and ensure fair allocation.

\T{Designing datacenter switches.} Could our switch-fabric designs fit datacenter switches? 
Datacenter switches rely on the same switch-fabric principles, so it looks like a natural extension. 
In fact, since they use less buffering than internet routers,  they may rely on the \arch architecture with even less power consumption, since they would use \HBMs with much less HBM capacity. However, (i)~\textit{latency} is more critical in datacenter networks. Thus, the \HBM may need to be modified to rely on smaller frames. In addition, (ii)~\textit{switch radix} may be a more significant concern in datacenters, with a finer output granularity than in core routers. 
Alternatively, a solution to both problems could be to implement \textit{an \arch architecture that relies on smaller commercial switches as chiplets}. We note that Broadcom has started shipping the Tomahawk 6 switch chip~\cite{tomahawk-reuters} and the Jericho 4 switch chip~\cite{jericho}, both relying on a chiplet architecture.

\T{Wasted internal traffic.}
{The ability to scale routers by 
1-2 orders of magnitude can  save a significant fraction of the current WAN capacity that is devoted to internal traffic needed to interconnect smaller routers.}

\section{Conclusion}
\label{sec:conclusion}
This paper introduces a router-in-a-package architecture that combines in-package optics with HBM4-based shared-memory switches to build a petabit-per-second Internet router within a single package. At the top level, the Split-Parallel Switch (SPS) spatially distributes fibers across several parallel HBM switches with only a single OEO conversion. Internally, each HBM switch implements a Parallel Frame Interleaving (PFI) mechanism that aggregates packets into large frames, exploits ultra-wide HBM interfaces with staggered bank interleaving, and uses simple cyclical crossbars to emulate an ideal output-queued shared-memory switch with 100\% throughput guarantees under admissible traffic.

\T{The road ahead.} What is the bottleneck to further scaling our router designs? 
\cref{sec:analysis} shows that the footprint of our reference design is only a small fraction of the surface area of a panel-scale substrate, so \textit{it is not area-bound}.
Also, 
\textit{it is not yet I/O-bound} since we have assumed {$40\unit{Gb/s}$} data rates per wavelength, whereas $112\unit{Gb/s}$ data rates per wavelength are already possible with PAM4 modulation~\cite{LM3}, and current panel substrates with $500\unit{mm}$ panel edges can support the attachment of many more fiber ribbons. In fact, the main limiting factor in our router-in-a-package design is \textit{power and heat dissipation}. 
As noted in \cref{sec:analysis}, the estimated power of our reference design is just above half that of a commercial wafer-scale processor for which existing state-of-the-art power delivery and cooling mechanisms are available~\cite{Kennedy24, Kundu25, Mujtaba21}. 

Going forward, future HBMs~\cite{HBM-nextgen, Hsu25a, Hsu25b} should require less power per bit, which is significant since HBM accounts for $40\%$ of our overall power.
However, 
\textit{the processing chiplets, with {50\%} of power, could become the next significant bottleneck in scaling routers.} 
Could operators reduce their processing needs if this increases their router capacity? Recent suggestions, such as source routing in SD-WANs~\cite{source-routing} and segment routing~\cite{segment-routing}, may lead the way to future simpler processing.

\T{Acknowledgments.} We thank Sylvia Ratnasamy, Scott Shenker, Alex Krentsel, Sarah McClure, Shai Cohen and Yoav Levi for their insightful discussions and comments. This work was partly supported
by the Louis and Miriam Benjamin Chair in Computer Communication Networks and NSF Award No. 2410053. Support for the Anonymized Two-Way Traffic Packet Header Traces 100G dataset~\cite{caida2025passive100g} is provided by the Center for Applied Internet Data Analysis (CAIDA) at the University of California San Diego.
\par
\textit{This work does not raise any ethical issues.}
\bibliographystyle{ACM-Reference-Format}
\bibliography{mybib}

\appendix

\section{Additional architecture analysis} \label{sec:add_architecture} 

In this section, we extend \Cref{sec:HBM} with additional architecture analysis.

\T{Frame and segment size.} 
As discussed in \Cref{sec:HBM}, the group of $4$ HBMs together provides $T = 4 \cdot 32 = 128$ channels across an ultra-wide $1,024\unit{ bytes}$
interface.
To take advantage of this ultra-wide interface with $T = 128$ channels, we choose a frame size of $K = T \cdot S$ bytes, where $S$ is the size of a {\em segment} that gets written into or read from each channel.
To determine the segment size $S$, we have four factors to consider, as partly discussed in \Cref{sec:HBM-process} and \Cref{sec:HBM}:

\Ts{Burst multiple.} First, once a row has been activated, consecutive columns can be read out (or written into) in burst-mode. In particular, HBM4 has a {\em burst-length} of $8$ and a channel width of $64\unit{ bits}$,
which means that each read/write command transfers 
$8 \cdot 64 = 
512\unit{ bits} = 64\unit{ bytes}$ each time. Therefore, the segment size $S$ should be a multiple of these bursts. 

\Ts{Activation limit.} 
Second, the concurrent activation of rows is power- and current-intensive. To prevent the memory from drawing too much instantaneous current, memory specifications (including HBM4~\cite{jedec}) typically impose a four-activation window 
constraint that limits at most four activated banks within a rolling time window of $t_\mathrm{FAW}$. Therefore, the segment size $S$ should be long enough to satisfy this constraint.

\Ts{Row divisor.} Third, the segment size should also be a unit fraction of the row size, so that a bank row can be evenly divided into segments.

\Ts{Minimum size.} Last, 
larger segments lead to larger frames, which increases the frame formation delay and the SRAM buffering requirement,
so smaller segments are preferred.

As discussed in \Cref{sec:HBM},  we choose $S = 1\unit{KB}$ for our reference to satisfy all these above constraints.

\T{Head SRAM.} As discussed in \Cref{sec:HBM}, the role of the head SRAM is to store frames read from the HBM group for departure via the corresponding HBM swtich outputs.
Like the tail SRAM, the head SRAM is organized into $N$ SRAM modules, each storing segments retrieved from a group of $T/N = 128/16 = 8$ channels. The HBM controller is responsible for matching the HBM's 10 Gb/s bit rate with the 2.5 GHz SRAM clock, with the corresponding deserialization of the $512$ bits from the HBM to the $2,048$-bit SRAM width.

Regarding the $N \times N$ round-robin crossbar on the output side that connects the $N$ SRAM modules in the head SRAM to the $N$ output ports,
it also implements a round-robin connection pattern, which 
can simply be implemented with 1-dimensional multiplexors with cyclic rotations (no scheduling). %
Like the input-side crossbar, this output-side crossbar can also equivalently be implemented by an $N \times N$ mesh with wire-splitting. 

\T{Frame interleaving cycles.}
The frame interleaving cycle shown in \Cref{fig:intuition-N} comprises the writing of $L$ frames from the tail SRAM into $L$ consecutive banks across $T$ channels of the HBM group, followed by the reading of $L$ frames from $L$ consecutive banks across $T$ channels of the HBMs to the head SRAM for transfer to the output ports. There is a small write-to-read transition delay ($t_\mathrm{WTR}$) between the frame write and read phases, and a similar read-to-write transition delay ($t_\mathrm{RTW}$) between the frame read and write phases. The total frame interleaving cycle duration is $(2L \cdot t_S + t_\mathrm{WTR} + t_\mathrm{RTW})$, where $t_S$ is the time for writing or reading a segment of size $S$ in burst-mode. Relative to the $2L$ accesses, $(t_\mathrm{WTR} + t_\mathrm{RTW})$ is about $0.05\%$ in duration, which is negligible. We define this pattern of $2L$ accesses with $t_\mathrm{WTR}$ and $t_\mathrm{RTW}$ as our baseline ``100\% throughput.''
Finally, HBM4~\cite{jedec} provides independent single-bank refresh operations that can be readily hidden between precharge and activation %
when needed, which does not affect our frame interleaving cycle time.

\T{Rate matching.} Each stage of the HBM switch architecture is designed to have matching memory speeds. In particularly, the $N$ input port SRAMs, the $N$ tail SRAM modules, the $B$ HBMs, and the $N$ head SRAM modules are all dimensioned to support $N \cdot P = 40.96\unit{ Tb/s}$ of writing and $N \cdot P = 40.96\unit{ Tb/s}$ of reading, for a combined access rate of $81.92\unit{ Tb/s}$. The two round-robin crossbars are dimensioned to support these data rates.

\section{Proofs}
\label{sec:proofs}

\begin{proof}[Proof of \cref{thm:bounded} (SRAM size)] We will prove the following SRAM upper bounds for the various system components: \\
    \noindent \emph{(i)}~The total SRAM size at input ports is at most $N\cdot (N\cdot \k+\k-N)$.\\
    \noindent \emph{(ii)}~The total tail SRAM size is at most $N\cdot \K+\K-N\cdot \k$.\\
    \noindent \emph{(iii)}~The total head SRAM size is at most $\frac{N+1}{2}\cdot \K$. \\
    \noindent \emph{(iv)}~The total output port SRAM size is at most $2 N\cdot \k$. 

    Let's now prove component by component:

\textit{(i)}~Consider input port $i\in\sbrac{1,N}$, and let's prove that its input port SRAM size is at most $N\cdot K$, yielding the result. 
We focus on the periodic time slots where input $i$ is getting connected to the top tail SRAM. Input $i$ can receive and send at most $K$ bits between these time slots. We distinguish three cases. (1)~If $i$ does not have a full frame at such a time slot, then all of its $N$ frame buffers have at most $K-1$ bits. Since it receives at most $K$ bits by the next such time slot, this satisfies the result till the next time slot. (2)~Else, if it is the first time slot with a full frame following a time slot without one, then since it could only have grown by at most $K$ bits in the last period, its size is at most $N\cdot (K-1) + K = N \cdot K$, also satisfying the result. From then on, $i$ is work-conserving and writing full frames, so its size cannot increase anymore in any contiguous period of full frames, and the result is satisfied. In particular, it holds (3)~at the periodic time slots where it is getting connected to the top tail SRAM. When it finishes writing all of its full frames, it is connected again to the top tail SRAM and we get back to case (1).

\emph{(ii)}~The proof follows similar lines. We focus on the periodic time slots where the \textit{frame interleaving cycles} start. The tail SRAM can only receive and send up to $L$ frames of size $K$ bits each between these time slots. 
Let's prove that its number of such frames does not exceed $N \cdot L$. We distinguish three cases. (1)~If the tail SRAM has at most $N\cdot (L-1)$ frames at such a time slot, then since it receives at most $L$ frames by the next such time slot, this satisfies the result till the next time slot. (2)~Else, if it is the first time slot with strictly more than $N\cdot (L-1)$ frames, then since it could only have grown by at most $L$ frames in the last period, its size is at most $N\cdot (L-1) + L = N \cdot L$, also satisfying the result. Moreover, since there are at least $N\cdot (L-1)+1$ frames, by the pigeonhole principle, at least one output holds at least $L$ frames. Since for each output, the tail SRAM writes into sequential banks (mod~$L$), this means that the tail SRAM has at least one frame for each bank and can write into all $L$ banks, \ie write at maximum rate. From then on, the tail SRAM size cannot increase anymore, and its size stays upper-bounded by $N\cdot L$ until the next time slot where the {frame interleaving cycle} starts. This keeps holding at (3)~all such time slots with more than $N\cdot (L-1)$ frames.

\emph{(iii-iv)}~At each head SRAM and each output port SRAM, we just need to buffer one frame of size $K$ bits per output, and there are $N$ outputs.

Finally, we sum up all SRAM  upper bounds to conclude.
\end{proof}

\begin{proof}[Proof of \cref{thm:mimic-frames} (Mimicking with frames only)]
Assume that only fully-formed frames arrive, both at the shared-memory switch and at the HBM switch. We saw that all SRAM sizes are bounded in the HBM switch. In addition, frames  keep getting serviced in a FIFO manner in each input port SRAM and in each per-bank queue at the tail SRAM. 
Thus, an incoming frame necessarily joins the HBM switch in a bounded time. Likewise, a frame departing the shared HBM memory group also joins the output port SRAMs in a bounded time, with a guaranteed service that cannot be preempted by frames destined to other outputs. Therefore, any incoming frame reaches the main shared memory and later departs the switch within a bounded delay difference from the ideal shared-memory switch.    
\end{proof}

\begin{proof}[Proof of \cref{thm:full-throughput} (100\% throughput for HBM)]
Consider some cell arrival process with a cumulative number of arrivals $\brac{A(t)}$ by time $t$. Given this arrival process, assume that by time $t$, $\DSM$ (resp. $\DHBM$) cells would depart the shared-memory switch (resp. the HBM switch). 
\Cref{thm:bounded} proves that a bounded number of the $\DHBM$ cells are stuck in the input port SRAMs before becoming frames. 
Once frames depart the input port SRAMs, \cref{thm:mimic-frames} proves that a bounded number of additional frames within the $D(t)$ are still not serviced by the HBM switch. 
Thus, we are guaranteed that there is some constant $B$ such that for any $t$, the number $\DHBM$ of cells that have departed the HBM switch satisfies 
$\DHBM \geq \DSM-B.$ 
Since $\lim_{t \to \infty} \frac{B}{t} = 0,$ this implies that both switches have the same throughput.
In particular, when $\brac{A(t)}$ is admissible, let the queue size of the shared-memory (resp. HBM) switch be $\QSM=A(t)-\DSM$ (resp. $\QHBM=A(t)-\DHBM$). Then $\QHBM-\QSM=\DSM-\DHBM \leq B$. Since $\brac{A(t)}$ is admissible, there is a constant $B'$ such that for any $t$, $\QSM \leq B'$. It follows that $\QHBM \leq B+B'$, 
thus the HBM switch can also service $\brac{A(t)}$ with a bounded queue size. The HBM switch achieves the same 100\% throughput for admissible traffic as the shared-memory switch.
\end{proof}

\begin{proof}[Proof of \cref{thm:global-full-throughput} (100\% throughput for \arch)]
No matter the pseudo-random assignment to the HBM switches, the TM at each HBM switch is the same. Therefore, for each given output $j$, the TM column sum at each HBM switch is the same. Hence, it equals the column sum for output $j$ in the parallel-switch TM scaled down by $1/H$. This implies that the smaller switch is just a slowed down version of the larger switch. For each $j$, a shared-memory switch would achieve 100\% throughput, and therefore by \cref{thm:full-throughput}, so would each HBM switch. This implies that the parallel switch guarantees 100\% throughput as well.
\end{proof}

\begin{proof}[Proof of \cref{thm:mimic-speedup} (Mimicking with speedup)]
Let $n>0$ be such that $1/n<\epsilon$, so $1+\epsilon>1+1/n=(n+1)/n.$ During any period of $n$ SRAM slots, the speedup enables the SRAMs to operate $n+1$ times. Assume that in every such period, they operate $n$ times by servicing regular frames, without any change for regular frames; and once by servicing marked frames at marked time slots, which will be defined below. 
Note that regular frames do not see any worsening in their service times.

Once per period, the HBM switch picks the VOQ with the oldest bit among all the input port SRAMs (\ie the bit that has waited the longest in the input port SRAMs), forms a frame using padding, and marks the frame. Then, it slowly services this marked frame in the next periods using the marked slots only, forming a pipeline, until it reaches the output. Thus, all widow cells are guaranteed to leave within a bounded amount of time. In addition, \cref{thm:mimic-frames} proves that the service times of the regular frames are within a bounded delay of their services times in a shared-memory switch. Combining both results yields the mimicking theorem.  
\end{proof}

\section{Synthetic WAN Flow Generation}
\label{sec:gen}

\T{WAN flow generation.}
We assume that in the WAN, the flow rate distribution is log-normal~\cite{Zhang2002}. Traces from 2015-2016 show median flow rates around $5\unit{kbps}$ and 99th percentiles around $600\unit{kbps}$, doubling every 6-7 years~\cite{Bauer2021}. Thus, we triple these values to set the log-normal distribution parameters ($\mu=-4.2$ so $e^\mu=0.015\unit{Mbps}$, and $\sigma=2.06$). 

\T{WAN TM generation.}
Next, we use a Zipf distribution with exponent $s$ to represent the non-uniformity of the output destinations~\cite{elhanany2002analysis}.
Given a traffic load $\rho \leq 1$ and a Zipf exponent $s$, we constrain the total load at each input wavelength to be at most $\rho$, and the total load at the $k$-th output to be at most $\rho/k^s$. %
For example, in the uniform case with $s=0$, all output loads are also at most $\rho$. Each time we generate a new flow rate, we assign it an input wavelength uniformly at random and a random output following the Zipf distribution. We insert it only if all our constraints are satisfied. We stop the packing process once 30 consecutive flows cannot be inserted. Note that a single TM can involve packing up to some $10^{10}$ flows in our evaluation.

\end{document}